\documentclass[a4paper,british,11pt]{scrartcl}
\usepackage{float}
\usepackage{amsmath, amssymb}

\usepackage{placeins}

\usepackage[colorlinks=true, linkcolor=blue, urlcolor=blue, citecolor=blue, breaklinks=true]{hyperref}
\usepackage[numbers,sort&compress]{natbib}

\usepackage[table]{xcolor}

\usepackage{mathtools}
\usepackage{stmaryrd}
\usepackage{color,framed}
  \usepackage{amsthm}
\usepackage[capitalise]{cleveref}
\usepackage{tikz}
\usepackage[USenglish]{babel}
\usetikzlibrary{decorations.pathmorphing}
\usepackage{pgfplots}
\usepgfplotslibrary{groupplots} 
\pgfplotsset{compat=1.18}
\usepackage{svg}
\usepackage{nicefrac}
\usepackage{booktabs}

  \newcommand{\optDisplay}[1]{\[#1\]}

  \newtheorem{theorem}{Theorem}[section]
  \newtheorem{lemma}{Lemma}[section]
  \newtheorem{corollary}{Corollary}[section]
  
  \newtheorem{claim}{Claim}[section]
  \crefname{claim}{Claim}{Claims}
  
  \newtheorem{proposition}{Proposition}[section]
  
  \theoremstyle{definition}
  \newtheorem{definition}{Definition}[section]

  \newtheorem{example}{Example}[section]

  \newtheorem{remark}{Remark}[section]

\AddToHook{cmd/appendix/before}{\crefalias{section}{appendix}}

\newcommand{\N}{\mathbb{N}}
\newcommand{\zeroto}[1]{\llbracket #1 \rrbracket}

\newcommand{\R}{\mathbb{R}}
\newcommand{\Q}{\mathbb{Q}}
\newcommand{\WE}{\textnormal{WE}}
\newcommand{\corresp}{\Gamma}
\newcount\colveccount
\newcommand*\colvec[1]{
        \global\colveccount#1
        \begin{pmatrix}
        \colvecnext
}
\def\colvecnext#1{
        #1
        \global\advance\colveccount-1
        \ifnum\colveccount>0
                \\
                \expandafter\colvecnext
        \else
                \end{pmatrix}
        \fi
}

\DeclareMathOperator*{\argmin}{\textnormal{arg\,min}}
\DeclareMathOperator*{\argmax}{\textnormal{arg\,max}}

\DeclareMathOperator{\Jac}{J}

\DeclarePairedDelimiter{\norm}{\lVert}{\rVert}

\usepackage[textwidth=2cm,textsize=scriptsize]{todonotes}

\usepackage[normalem]{ulem}
\newcommand\redsout{\bgroup\markoverwith
{\textcolor{red}{\rule[0.5ex]{2pt}{1pt}}}\ULon}

\newcommand{\sub}{\textsubscript}

\newcommand{\diff}{\,\mathrm{d}}
\newcommand{\scalprod}[2]{\langle #1 , #2 \rangle}

\newcommand{\feasFlows}{\mathcal{F}}
\newcommand{\setCommodities}{I}

\newcommand{\plotUnits}[1]{#1}

\NewDocumentCommand{\plotInstanceTikz}{ O{ht} O{1} O{} m m }{
  {\footnotesize
    \begin{tikzpicture}
      \begin{groupplot}[group style={group size=2 by 1, horizontal sep=3cm}, width=\plotwidth, height=\plotheight, grid=major]
        \nextgroupplot[xlabel={\begin{tabular}{c}Price\plotUnits{\\[-0.5ex]\tiny (in min/#3km)}\end{tabular}}, ylabel={\begin{tabular}{c}Total Emission\plotUnits{\\\tiny (in km)}\end{tabular}}]
          \addplot[-] table [x expr=\thisrow{price}*#2, y=total_consumption, col sep=comma] {data/#4.csv};
          \carbonPricingLowerBound{#4}{#2}
      \nextgroupplot[xlabel={\begin{tabular}{c}Price\plotUnits{\\[-0.5ex]\tiny (in min/#3km)}\end{tabular}}, ylabel={\begin{tabular}{c}Total Travel \\ Time\plotUnits{ \tiny (in min)}\end{tabular}}]
        \addplot[-] table [x expr=\thisrow{price}*#2, y=total_travel_time, col sep=comma] {data/#4.csv};
      \end{groupplot}
    \end{tikzpicture}}
}

\NewDocumentCommand{\plotInstance}{ O{ht} O{1} O{} m m }{
  \begin{figure}[#1]
    \centering
    \plotInstanceTikz[#1][#2][#3]{#4}{#5}
    \vspace{-1em}
    \caption{#5}\label{fig:#4}
  \end{figure}
}

\newcommand{\carbonPricingLowerBound}[2]{%
  \def\readMinimumEmission##1-carbon-pricing\relax{%
    \pgfplotstableread[col sep=comma]{data/##1-min-emission.csv}\minimumEmissionTable
  }%
  \readMinimumEmission#1\relax
  \pgfplotstablegetelem{0}{total_consumption}\of\minimumEmissionTable
  \edef\minimumEmission{\pgfplotsretval}%
  \addplot[black, dashed, no marks, forget plot]
    table[x expr=\thisrow{price}*#2, y expr=\minimumEmission, col sep=comma]
    {data/#1.csv};
}

\usepackage{thmtools}

\newenvironment{subproof}{\begin{proof}}{\end{proof}}

\newenvironment{applications}{\begin{description}}{\end{description}}

\newtheoremstyle{mainresultstyle}%
  {3pt}%
  {3pt}%
  {\normalfont\itshape}%
  {}%
  {\normalfont\scshape}%
  {.}%
  { }%
  {}%

\theoremstyle{mainresultstyle}
\newtheorem{mresult}{Main Result}
\crefname{mresult}{Main Result}{Main Results}

\newenvironment{result}{%
    \begin{framed}%
    \begin{mresult}~\\ \vspace{-0.5cm}%
}{%
    \vspace{-0.5cm}%
    \end{mresult}%
    \end{framed}%
}

\usepackage{authblk}

\title{Carbon Pricing in Traffic Networks}
\author{Svenja M. Griesbach}
\affil{Universidad de Chile}
\author[2]{Tobias Harks}
\affil[2]{University of Passau}
\author{Max Klimm}
\affil{TU Berlin}
\author[2]{Michael Markl}
\author{Philipp Warode}
\affil{Humboldt-Universität zu Berlin}

\usepackage[left=3cm, right=3cm, top=3cm, bottom=3cm]{geometry}

\date{}

\begin{document}

\maketitle

\begin{abstract}
We study how externality pricing can steer traffic toward Wardrop equilibria respecting prescribed budgets.
Our multi-commodity model encompasses carbon pricing, cordon pricing, and tradable credit schemes.
We characterize the flows attainable when one common price may be assigned to each externality.
We show that for continuous travel times and flow-independent externalities, every feasible externality budget is implementable.
For flow-dependent externalities, smoothness bounds quantify the necessary budget relaxation; this links budget implementability to the price of anarchy.
If travel times and externalities admit convex potentials, we give conditions for constructing prices implementing a given budget through convex optimization and specify explicit prices ensuring that relaxed budgets are respected by every equilibrium.
As, for a single separable constant externality, total externality is nonincreasing in its price, this enables us to compute the whole price--equilibrium curve in finite time for separable, continuous, strictly increasing, piecewise affine travel times.
Although this curve may have exponentially many pieces, the minimum price implementing a feasible budget is computable in polynomial time.
The same machinery computes the range of market-clearing prices for tradable credit schemes.
A computational study on several realistic traffic networks illustrates emission reductions from carbon pricing and the effects of cordon pricing on emissions within the cordon.
\end{abstract}

\clearpage

\tableofcontents

\clearpage

\bibliographystyle{plainnat}

\section{Introduction}\label{sec:intro}
 
Reducing carbon emissions is a key goal of the European Union's Emissions Trading System~(ETS). A central measure is to associate a carbon price with the emission of greenhouse gases with effects equivalent to a tonne of CO\sub2.
While the objective of the original ETS-I was to bring down emissions in power generation and energy-intensive industries, the new emissions trading system ETS-II, which is expected to start in 2028, will 
target the reduction of emissions caused by road transportation.
This focus on traffic is triggered by the fact that transport was responsible for about a quarter of the EU's total CO\sub2 emissions in 2018, of which~72\% came from road transportation~\citep{EU20}.
The imposition of a carbon price in traffic networks clearly has the potential to change the resulting equilibrium flow patterns substantially; in fact, this drastic influence on the flow is deliberate and the reason why the pricing is implemented in the first place. However, to assess the effect of carbon pricing, it is imperative to understand how traffic flows change under a carbon price. This is the question studied in this paper. {}

More specifically, we study the influence of a carbon price on the
resulting equilibrium flow, assuming the well-established 
Wardrop flow model~\citep{CorreaS2011,Wardrop52}.
In this model, we are given a directed graph~$G = (V,E)$ modeling road segments in a street network.
There is a finite set~$I$ of origin--destination pairs, called \emph{commodities}. Every commodity
has a \emph{demand}~$d_i$ specifying the amount of
flow that needs to be sent from the respective origin~$s_i \in V$ to the destination~$t_i \in V$.
Commodities represent large populations
of players, each controlling an infinitesimally small amount of flow of
the entire demand; such players are also called \emph{nonatomic}.
The travel time a player of commodity~$i$ experiences when traversing path~$p$ from~$s_i$ to~$t_i$ is given by a travel time function~$\tau_{i,p}(x)$ that may depend on the entire multi-commodity path flow~$x$ (nonseparable travel time functions) or on the edge-flow only (separable travel time functions).
We assume that every player acts selfishly and routes their flow along a
minimum-latency path from its origin to the destination; such a flow is called a \emph{Wardrop equilibrium}.

Wardrop equilibria may be inefficient in the sense that the emerging traffic patterns do not minimize the total travel time; this has been known qualitatively for more than a century \citep{Knight1924,Pigou20} and this effect has been studied quantitatively more recently as the \emph{price of anarchy} \citep{Rough2002,Roughgarden02}.
Based on this observation, there is a substantial literature, discussed in detail in the related work section, studying the use of road tolls with the goal of reducing the total travel time.  
The predominant part of these works, however, assumes that potentially every edge of the network is priced, see e.g.~\citet{Beckmann56,Bergendorff97congestion,Cole06taxes,DialpartI,DialpartII,Fleischer04,Karakostas04,Larsson99,Yang04}.
Other works consider situations where only a subset of the edges can be priced, see e.g.~\citet{BonifaciSS11,HarksKKM2015,Hoefer08,JelinekKS14,Verhoef02}, but these models assume that the prices on these edges can be chosen independently.
Neither model captures the reality of carbon pricing where a \emph{single} price per CO\sub2 equivalent emission is fixed, and different edges and commodities have different characteristics regarding their respective emissions.

To address this shortcoming, in this work, we introduce a general model with a finite set of externality classes~$J$ such that when commodity~$i$ uses a path~$p$ from~$s_i$ to~$t_i$ in the network, this induces~$g_{i,p,j}(x) \in \mathbb{R}_{\geq 0}$ units of externality of class~$j \in J$.
We consider the task of a governing agency that has the power to impose a price~$\lambda_j \in \R_{\geq 0}$ per unit of externality class~$j$ and assume that a Wardrop equilibrium emerges where each player of commodity~$i$ evaluates the cost of path~$p$ as~$\tau_{i,p}(x) + \sum_{j \in J} \lambda_j \cdot g_{i,p,j}(x)$.

This allows us to model the following applications:
\begin{applications}
    \item[Carbon pricing.] When~$J = \{1\}$ and~$g_{i,p,1}$ is the amount of CO\sub2 equivalent emissions of a unit flow of commodity~$i$ traversing path~$p$, we obtain carbon pricing. 
    Note that we allow multiple commodities~$i, i' \in I$ to have the same origin~$s_i = s_{i'}$ and destination~$t_i = t_{i'}$, so that also different vehicle classes with different emission profiles can be modeled (e.g., cars with combustion engines, EVs, and trucks). This further allows us to model heterogeneous populations of traffic participants who evaluate the trade-off between the time costs~$\tau_{i,p}(x)$ and the monetary costs $\sum_{j \in J} \lambda_j \cdot g_{i,p,j}(x)$ differently.

    \item[Pricing multiple pollutants.] Traffic simulations, in fact, not only model greenhouse gas emissions but also trace the emissions of multiple additional pollutants, such as particles of different sizes, black carbon, and nitrogen oxides \citep{KaddouraEM22,GableMN22}. A straightforward generalization of the carbon pricing setting with~$|J| > 1$ allows introducing separate prices for the emission of these exhausts.

    \item[Road pricing.] Let~$J = E$ and let for each commodity~$i$ a parameter~$\alpha_i \in \R_{\geq 0}$ be the time--money trade-off of players of that commodity. Setting~$g_{i,p,e} \equiv \alpha_i$ when~$e \in p$ and~$g_{i,p,e} \equiv 0$, otherwise, we obtain road pricing for heterogeneous players as studied by \citet{Cole06taxes,Fleischer04,Karakostas04,Yang04} as a special case of our model.

    \item[Road pricing with bounded support.]
    Let~$T \subseteq E$ be a subset of edges where road pricing is possible.
    As above, defining~$J = T$ and~$g_{i,p,e} \equiv \alpha_i$ if~$e \in p \cap T$ and $g_{i,p,e} \equiv 0$ otherwise, we obtain the road pricing setting with bounded support as studied by \citet{HarksKKM2015,Hoefer08,Verhoef02}.

    \item[Cordon pricing.]
    Pricing access to city centers is a common measure for regulating traffic; see~\citet{SmallG98} for a survey on real-world implementations of this strategy.
    In the basic model of cordon pricing, access to the nodes of a city center~$C \subseteq V$ is charged with a single price~$\lambda$. To model this{} situation, we set $J = \{1\}$. There are two natural variants of the model. In the first,  every edge in
    \begin{align*}
    \delta^-(C) = \{(u,v) \in E : u \in V \setminus C, v \in C\}
    \end{align*}
    is priced so that a path that enters the city center multiple times pays the access fee multiple times.
    This payment scheme has been used, e.g., for the cordon pricing in Singapore~\citep{SmallG98}.
    In this case, we define~$g_{i,p,1} \equiv | \delta^-(C) \cap p|$. If access to the city center has to be paid only once for a trip, we may set~$g_{i,p,1} \equiv 1$ if~$\delta^-(C) \cap p \neq \emptyset$ and~$g_{i,p,1} \equiv 0$ otherwise, instead. This payment scheme is a realistic model, for example, for the pricing scheme in London, where access to the city center is charged per day \citep{PeirsonV08}.
    For both models, heterogeneous players with different time--money trade-offs may be modeled as above.

    \item[Tradable credit schemes.]
    As an alternative to a central authority setting a price for a given externality, \citet{Yang2011} introduced so-called tradable credit schemes, where the governing agency issues a fixed number~$B$ of credits (of a single externality class, $|J|=1$) and sets a constant credit charge~$g_e$ for each edge~$e\in E$.
    The issued credits can then be freely traded among the travelers.
    As travelers minimize their combined travel time and credit cost, a market equilibrium emerges consisting of a market-price~$\lambda$ (a price per credit) and a market-equilibrium flow~$x$.
    As we show{} (cf.~Result~\ref{res:affine}), these market equilibrium flows are exactly the flows that can be implemented by a specific range of (market equilibrium) prices.
\end{applications}

For our general model, encompassing the concrete applications above, we are interested in the following fundamental implementability questions, where we say that a feasible flow~$\smash{x\in \feasFlows}$ is \emph{implemented} by~$\lambda$, if~$x$ is a Wardrop equilibrium with respect to effective costs defined as the sum of travel times~$\tau_{i,p}(x)$ and prices~$\smash{\sum_{j \in J} \lambda_j g_{i,p,j}(x)}$ paid.
In this case, we write~$x\in\WE(\lambda)$.
We say that budgets~$\smash{B\in\R_{\geq0}^J}$ are \emph{implementable}, if there are prices~$\lambda$ and a corresponding Wardrop equilibrium~$x\in\WE(\lambda)$ that respects these budgets, i.e., $\sum_{i,p} g_{i,p,j}(x) \, x_{i,p} \leq B_j$ for all~$j \in J$ where~$x_{i,p}$ is the flow of commodity~$i$ on path~$p$ and the sum is over all commodity--path pairs.

We are specifically interested in the following questions: 
\begin{enumerate}
    \item  Which feasible flows~$x\in \feasFlows$ can be implemented by price vectors~$\smash{\lambda \in \R^J_{\geq 0}}$?
    \item Which feasible externality budgets~$B\in\R_{\geq0}^J$ can be implemented, and when can prices ensure that every equilibrium respects them?
    If a feasible budget is not implementable, how much must it be relaxed?
    \item For a single externality class, can we efficiently compute a minimal price~$\lambda$ such that there is~$x\in\WE(\lambda)$ with~$G(x)\leq B$, where~$G(x)$ denotes the induced externality?
    \item For a single externality class, can we describe and efficiently compute an entire equilibrium curve~$\lambda\mapsto x(\lambda)$ such that $x(\lambda)\in\WE(\lambda)$? How does this curve depend on the network structure? 
\end{enumerate}
\subsection{Our contribution}
The above questions are getting increasingly complex, and in the following, we will
address these questions in a decreasing degree of generality with respect to the cost structure.
We first characterize implementability for general nonseparable travel times, then derive constructive results when convex potentials are available, and finally give exact algorithms for separable piecewise affine travel times.

\medskip
\paragraph{General travel time functions.}
In \Cref{sec:general}, we allow the travel time of each path and commodity to depend on the entire path flow vector.
This class of travel time functions allows, for instance, the modeling of spill-back effects and thus leads to a more realistic static equilibrium model.
Our first main results can be summarized as follows, assuming all travel time and externality functions are continuous:
\begin{result}
    \begin{enumerate}
        \item A flow is implementable \emph{if and only if} it solves an associated linear program (\Cref{thm:implementability-lp-characterization}).
        \item For constant externalities, \emph{every feasible budget} is implementable (\Cref{thm:all-feasible-budgets-implementable-kakutani}).
        \item For $(\alpha_j,\beta_j)$-smooth externalities, scaling each strictly feasible budget~$B_j$ by~$\alpha_j/(1-\beta_j)$ makes it implementable (\Cref{thm:smoothness-implementability}).
    \end{enumerate}
\end{result}
For constant externalities, we obtain the stronger guarantee that for given budgets, there are prices for which \emph{every} equilibrium respects the budgets when the budget vector is extremal or there is only one externality (\Cref{thm:pareto-implementable,cor:implementability-single-resources}).
We call such budgets \emph{strictly implementable}.

The proof of \Cref{thm:implementability-lp-characterization} uses a natural LP-formulation, where the objective is to minimize the total travel time (evaluated at the candidate flow) subject to budget constraints.
The dual variables of the budget constraints serve as the implementing prices.
This approach has been applied before by \citet{Fleischer04}, \citet{Karakostas04}, \citet{Yang04}, and \citet{HarksS23}.
The difference, however, is that in the above works, every edge can be potentially priced, whereas in our model, only 
a per-unit price of the corresponding induced externality can be priced.
For constant externalities, Kakutani's fixed point theorem yields an implementable flow respecting any feasible budget vector.

For flow-dependent externalities, even a minimum feasible budget need not be implementable.
We therefore introduce an \emph{implementability factor}, measuring the worst-case multiplicative relaxation needed to implement strictly feasible budgets.
A smoothness argument bounds this factor from above by the same expression used to bound the price of anarchy, while the price of stability provides a lower bound.
These bounds coincide for standard classes of separable nondecreasing externality functions with tight smoothness bounds.
Thus, bounds on the inefficiency of selfish routing also quantify the limitations of externality pricing.

These existence guarantees leave open how to compute implementing prices and a corresponding equilibrium.
In general, computing an equilibrium is PPAD-hard even for linear nonseparable travel times and parallel edges (\Cref{lem:ppad-hardness}).
This motivates the additional structure considered next.

\medskip
\paragraph{Convex potential functions.}
In \Cref{sec:potential}, we study travel time and externality functions that admit convex potentials.
This includes separable nondecreasing functions and allows us to characterize equilibria through convex optimization.
Under these convexity assumptions, we obtain the following results.
\begin{result}
    \begin{enumerate}
        \item For constant externalities, implementing prices and equilibria for every feasible budget can be obtained through a constrained convex program (\Cref{thm:implementability-for-potential-travel-time}).
        \item For a single $(\alpha,\beta)$-smooth externality, given a feasible flow strictly respecting a given budget, prices can be explicitly specified that strictly implement the budget relaxed by the smoothness factor $\frac{\alpha}{1-\beta}$.
        For multiple externalities, an additional factor of~$|J|$ suffices (\Cref{thm:explicit-smooth-prices}).
        \item For a single constant externality, total externality is nonincreasing in the price (\Cref{thm:monotonicity}).
    \end{enumerate}
\end{result}
    More specifically, we show in \Cref{thm:implementability-for-potential-travel-time} that a feasible budget is implementable if an associated convex program is feasible and admits a KKT multiplier vector.
    Compared to \Cref{thm:all-feasible-budgets-implementable-kakutani,thm:smoothness-implementability}, this adds a sufficient condition for the implementability of budget vectors for \emph{non-constant} (but potential) externality functions that provides theoretical access to computing a corresponding price vector and equilibrium flow through the use of convex programming methods.

    For separable externalities that are polynomial in the edge load, we derive degree-dependent budget relaxation factors (\Cref{cor:polynomials}).
    In particular, the smoothness argument gives the classical price-of-anarchy relaxation factor $1/(1-\beta(d))$ for polynomials of degree $d$, where $\beta(0)=0$ and $\beta(d)=d/(d+1)^{(d+1)/d}$ for~$d>0$.
    The constrained convex program yields a constructive guarantee with the larger factor~$d+1$.
    Our explicit price construction recovers the smoothness factor for a single externality and guarantees that every equilibrium respects the relaxed budget.
    A corresponding equilibrium can then be found by convex optimization.

    Finally, monotonicity for constant externalities permits binary search for a minimum implementing price.
    For non-constant externalities, a Braess-inspired network violates the monotonicity (\Cref{example:non-monotonicity}).

\medskip
\paragraph{Piecewise affine travel time functions.}
In \Cref{sec:affine}, we consider separable, continuous, strictly increasing, piecewise affine travel times and a single separable constant externality.
This class of travel-time functions allows approximating any nondecreasing continuous separable travel time function defined on a compact interval to arbitrary precision.
For rational input data, we obtain the following algorithmic results.
\begin{result}
\label{res:affine}
    \begin{enumerate}
        \item A continuous, piecewise affine function~$x(\lambda)$ mapping every~$\lambda \geq 0$ to a Wardrop equilibrium flow can be computed exactly in finite time; however, it may have exponentially many breakpoints (\Cref{thm:piecewise:computation}).
        \item For a feasible budget~$B$, the minimum price~$\lambda$ with~$G(x(\lambda))\leq B$ can be found in polynomial time (\Cref{thm:binary-search}).
        \item For a feasible tradable credit scheme, a market equilibrium and the interval of all market equilibrium prices can be computed in polynomial time (\Cref{cor:tradable-credit-scheme-market-price}).
    \end{enumerate}
\end{result}

    For the proof of \Cref{thm:piecewise:computation}, we first show that even though the Wardrop equilibria~$x(\lambda)$ need not be unique, there is a piecewise affine function~$x(\lambda)$ such that~$x(\lambda)$ is \emph{a} Wardrop equilibrium for price~$\lambda$.
    We provide an algorithm that computes the whole solution curve~$x(\lambda)$ in finite time.
    Given the uniqueness of the total externality for a fixed~$\lambda$, this algorithm allows optimizing more intricate objective functions, such as a weighted sum of travel times and the externality, over the set of all prices~$\lambda$.
    However, we show that there are instances where the number of affine segments is exponential in the encoding length of the input; a similar example has been given by \citet{KlimmWarode2021} for single-commodity Wardrop equilibria as a function of the demand.
    
    \Cref{thm:binary-search} is attained by a binary search algorithm that relies on the monotonicity result from \Cref{sec:potential}.
    Using Cramer's rule and Hadamard's inequality, we give a lower bound on the length of each segment of the function~$x(\lambda)$.
    This ensures that the binary search finds the minimal price for which the corresponding Wardrop equilibrium satisfies a given budget constraint in only polynomially many steps.

\paragraph{Computational study.}
In \Cref{sec:comp-study}, we examine carbon and cordon pricing in real-world traffic networks, using road length as a proxy for emissions and a Frank--Wolfe method to approximate equilibria at sampled prices.
For the networks tested, carbon pricing produces emission curves with an approximately convex shape; thus, even with comparatively low prices emissions can be reduced rather significantly.
For example, in the Berlin Center network, emissions fall by up to~$4.2\%$.
For permit-based cordon pricing in the same network, the number of users driving in the city center falls by up to~$2.4\%$ and emissions within the center by up to~$1.4\%$.
The study thus illustrates both the potential of direct emission pricing and the limitations of pricing access to a geographic area as a means of reducing emissions.

\subsection{Related work}
\label{sec:related-work}

\citet{Beckmann56} showed that when the travel-time functions are separable and nondecreasing, a Wardrop equilibrium can be computed by solving a convex optimization problem.
\citet{Dafermos71} showed that Wardrop equilibria with nonseparable edge travel times are the solution to a convex optimization problem when the Jacobian matrix of the travel time function is symmetric and positive semidefinite.
For the asymmetric case, \citet{Smith79existence} provided a variational inequality and studied uniqueness of equilibria.

There is a rich theory of road pricing in traffic networks that dates back to the early work of \citet{Pigou20} and \citet{Knight1924}.
\citet{Beckmann56} showed that charging the difference between marginal social costs and marginal costs imposes the system-optimal flow in the equilibrium.
\citet{Bergendorff97congestion}, \citet{Hearn98solving}, and \citet{Larsson99} showed that the set of road toll vectors that induce the system-optimal flow is a polyhedron.
\citet{DialpartI,DialpartII} devised algorithms that optimize an affine function over that polyhedron, such as the minimization of the total tolls collected from the traffic users.
\citet{Bai04}, \citet{Bai09}, and \citet{Bai10} studied the problem of finding a toll vector in the polyhedron with minimal support.

\citet{Verhoef02} studied the problem of computing road pricing schemes where only a given set of edges can be priced.
He proposed a heuristic and evaluated its performance on a small network.
\citet{Hoefer08} studied this problem by analyzing the computational complexity of finding a toll vector with a given support that minimizes the total travel time of the induced equilibrium. This question was further explored by \citet{HarksKKM2015}.
\citet{BonifaciSS11} and \citet{JelinekKS14} considered a generalization of the problem with an edge-dependent upper bound on the toll.
\citet{Cole06taxes}, \citet{Fleischer04}, \citet{Karakostas04}, and \citet{Yang04}
gave complete characterizations of Wardrop equilibrium flows that can be implemented by edge-based prices, even when players have heterogeneous preferences over travel times and pay prices.
\citet{HarksS23}
extended this characterization to hold for nonseparable and player-specific travel time functions. 

The cordon pricing problem has been mainly studied from a heuristic perspective.
\citet{MayLSS02} and \citet{Sumalee04} proposed genetic algorithms computing a cordon and a corresponding price.

There is vast literature on tradable credit schemes dating back to the works of \citet{Goddard1997} and \citet{Verhoef1997}.
\citet{Yang2011} first formalized this equilibrium model for separable and monotone travel time functions.
They provided sufficient conditions for the uniqueness of the market equilibrium price and studied whether the credit scheme can be chosen to attain desirable market-equilibrium flows, such as social optimum flows{} or so-called Pareto-improving flows.
    \citet{Wang2012} and \citet{Zhu2015} extended the model to allow for heterogeneous players with different value of time parameters; \citet{Nie2013} and \citet{Bao2019} consider tradable credit schemes for dynamic flow models such as Vickrey's bottleneck model.
    For a broader overview of the literature on tradable credit schemes, we refer the reader to the recent work of \citet{Liu2024}.

\section{The model}\label{sec:model}
For~$k \in \mathbb{N}$, let~$[k] \coloneqq \{1,\dots,k\}$ and let~$\zeroto{k} \coloneqq \{0,\dots,k\}$.
We are given a directed graph~$G=(V, E)$ with~$V$ the set of nodes and~$E$ the set of edges.
For a Boolean expression~$A$, we denote by $\mathbf{1}_A$ its indicator function which is $1$ if~$A$ is true, and~$0$ otherwise. 
Let~$I$ denote the set of commodities and let~$s_i \in V$ and~$t_i \in V$ denote the source and sink node of commodity~$i$, respectively.
The sink $t_i$ is assumed to be reachable from the source $s_i$ in $G$.
Furthermore, each commodity has a demand~$d_i > 0$.

Let~$\mathcal P_i$ denote the set of simple paths from~$s_i$ to~$t_i$ for commodity~$i$, and let~$\mathcal P$ denote the set of all pairs~$(i,p)$ with~$i \in I$ and~$p \in \mathcal P_i$.
A \emph{flow}~$x$ is a vector in~$\R^{\mathcal P}_{\geq 0}$ with~$x_{i,p}$ denoting the flow of commodity~$i$ on path~$p$.
We call~$x$ \emph{feasible} if it satisfies the demand, i.e., $\sum_{p\in\mathcal P_i} x_{i,p} = d_i$ for all~$i \in I$.
Let~$\feasFlows$ denote the set of feasible flows.
We are given a finite set of externality classes~$J$ together with a set of externality functions~$g_{i,p,j}: \feasFlows \to \R_{\geq0}$ for all~$i\in I$,~$p\in \mathcal P_i$, and~$j\in J$. 
Here,~$g_{i,p,j}(x)$ models the induced externality of externality class~$j$ for a single unit of flow of commodity~$i$ on path~$p$ under flow~$x$.
The total externality of class~$j$ for a flow~$x$ is given by \optDisplay{
  G_j(x) \coloneqq \sum_{(i,p)\in\mathcal P} g_{i,p,j}(x) \cdot x_{i,p}.
}
A \emph{cost function} is a function~$\smash{c:\mathcal F \to \R^{\mathcal P}}$ that assigns each path~$p$ of each commodity~$i$ a cost~$c_{i,p}(x)$ depending on a feasible flow~$x$.
\begin{definition}
  A feasible flow~$x$ is a \emph{Wardrop equilibrium} with respect to the cost functions~$c:\mathcal F \to \R^{\mathcal P}$ if for all~$i \in I$ and all paths~$p\in\mathcal P_i$ it holds that \[
    x_{i,p} > 0 \implies 
     c_{i,p}(x) \leq c_{i,q}(x) \text{ for all } q\in \mathcal P_i.
  \]
\end{definition}

In this work, we analyze Wardrop equilibria with respect to a parametrized cost function~$c^\lambda$ of the form
\optDisplay{
  c^\lambda_{i,p}(x) \coloneqq \tau_{i,p}(x) + \sum_{j\in J} \lambda_j \cdot g_{i,p,j}(x),
}
where~$\tau_{i,p}(x)$ is the travel time for commodity~$i$ on path~$p$, and~$\lambda_j$ denotes the price for inducing a single unit of externality class~$j\in J$.
For an externality price vector~$\smash{\lambda\in\R_{\geq0}^J}$,
  we denote by~$\WE(\lambda)$ the set of all Wardrop equilibria with respect to~$\smash{c^\lambda}$.
We further introduce the shorthands $\tau\coloneqq (\tau_{i,p})_{(i,p)\in\mathcal{P}}$ and $g_j \coloneqq (g_{i,p,j})_{(i,p)\in\mathcal{P}}$ for $j\in J$.

    \begin{definition}
      We introduce the following terminology:
      \begin{enumerate}
      \item We call a budget vector~$B\in\R_{\geq 0}^J$ \emph{feasible}, if there exists a flow~$x\in\feasFlows$ with~$G_j(x)\leq B_j$ for all~$j\in J$; in this case we say~$x\in\feasFlows$ \emph{respects}{} the budget~$B$.
      If all inequalities hold strictly, we say $B$ is \emph{strictly feasible} and $x$ \emph{strictly respects} $B$.
        \item A flow~$x\in \feasFlows$ is called \emph{implementable}, if there is a price vector~$\lambda\in\R_{\geq0}^J$ such that~$x\in \WE(\lambda)$.
        \item A budget vector~$B\in\R_{\geq0}^J$ is \emph{implementable},
        if there is an implementable flow~$x$ that respects~$B$.
        \item A budget vector~$B\in\R_{\geq0}^J$ is \emph{strictly implementable}, if~$B$ is implementable via some price vector~$\lambda\in\R_{\geq0}^J$ such that \emph{every}{} Wardrop equilibrium~$x\in\WE(\lambda)$ respects~$B$.
      \end{enumerate}
    \end{definition}
  To avoid confusion, note that \emph{strict implementability} implies that every equilibrium in $\WE(\lambda)$ fulfils \emph{weak} budget constraints; it does \emph{not} imply that there is an equilibrium in $\WE(\lambda)$ that fulfils the constraints \emph{strictly}.

\section{General travel time functions}
\label{sec:general}

We first consider general travel time functions~$\tau_{i,p}: \feasFlows \to \R$ and externality functions~$g_{i,p,j}: \feasFlows \to \R$.
For this general case, we derive a characterization of implementable flows.
More specifically, we show that a flow is implementable if and only if it minimizes a linear program in which the flow itself appears in the objective.
For \emph{budgets} we show that any feasible budget vector is implementable given that the travel time functions are continuous and the externality functions are constant.

\subsection{Implementability of Flows.}
    It is well-known (see \citealt{Smith79existence}) that the set of Wardrop equilibria with respect to cost functions~$c$ can be characterized as the solutions to the following variational inequality:
    \begin{equation}\label{eq:variational-ineq}\tag{VI}
      \text{Find~$x\in\feasFlows$ such that }
      \quad
      \scalprod{c(x)}{y - x} \geq 0 \quad \text{ for all } y\in\feasFlows.
    \end{equation}
    
    \begin{lemma}\label{lem:wardrop-eq-vi}
      A feasible flow~$x$ is a Wardrop equilibrium with respect to~$c$ if and only if it solves~\eqref{eq:variational-ineq}.
    \end{lemma}
    Browder's theorem (e.g., \citealt[Theorem~2]{Browder1968}) states that such a variational inequality always has a solution for continuous maps defined on a compact convex subset of~$\R^n$.
    Hence, for continuous cost functions~$c$, the existence of a Wardrop equilibrium is guaranteed.

    In our setting, we consider the parametrized cost function~$c^\lambda_{i,p} = \tau_{i,p} + \sum_j \lambda_j \cdot g_{i,p,j}$ and the set of \emph{implementable} flows, i.e., Wardrop equilibria for which there exists a price vector~$\lambda\in\R_{\geq0}^J$ such that~$x$ is a Wardrop equilibrium with respect to~$c^\lambda$.

Clearly, not every feasible flow is implementable.
Hence, we are interested in conditions characterizing the implementability of flows.
Similar to the approaches in \citet{Fleischer04} and in \citet{HarksS23}, we define
a linear program with externality constraints whose dual defines corresponding prices:
\begin{restatable}{theorem}{ThmImplementabilityLPCharacterization}\label{thm:implementability-lp-characterization}
  A flow~$u\in\mathcal F$ is implementable if and only if it solves the following linear program:
  \begin{equation}
    \begin{aligned}\label{lp:characterization-implementability}
      \min_{x\in\mathcal F} & \sum_{(i,p)\in\mathcal P} \tau_{i,p}(u) \cdot x_{i,p} \\
      \text{s.t.} & \sum_{(i,p)\in\mathcal P} g_{i,p,j}(u) \cdot x_{i,p} \leq G_j(u) \quad \text{for all } j\in J.
    \end{aligned}
  \end{equation}
\end{restatable}
\newcommand{\proofOfThmImplementabilityLPCharacterization}{
\begin{proof}
  The dual problem of the above linear optimization problem is given by
  \begin{align*}
    \max_{(\lambda, \mu)\in\R_{\geq 0}^J \times \R^I} \quad & \sum_{j\in J} (- G_j(u))\cdot \lambda_j + \sum_{i\in I}d_i\cdot \mu_i  \\
    \text{s.t.}\quad  & \sum_{j\in J} (-g_{i,p,j}(u))\cdot \lambda_j + \mu_i \leq \tau_{i,p}(u) \quad \text{for all } (i,p)\in\mathcal P,
  \end{align*}
  where~$\mu_i$ is the dual variable for the primal constraint~$\sum_{p\in\mathcal P_i} x_{i,p} = d_i$.
  Assume~$u$ is implementable, i.e., there exists~$\smash{\lambda\in\R_{\geq 0}^J}$ such that~$u$ is a Wardrop equilibrium with respect to~$\smash{c^\lambda}$.
  Clearly,~$u$ is feasible for the primal problem, and~$(\lambda, \mu)$ with~$\mu_i\coloneqq \min_{p\in \mathcal P_i} \tau_{i,p}(u) + \sum_{j\in J} \lambda_j \cdot g_{i,p,j}(u)$ is feasible for the dual problem.
  Then, as~$u$ is a Wardrop equilibrium with respect to~$\smash{c^\lambda}$, we have~$d_{i} \cdot \mu_i = \sum_{p\in\mathcal P_i} c^\lambda_{i,p}(u) \cdot u_{i,p}$.
  Summing this over all~$i\in I$, we obtain~$\smash{\sum_{i\in I} d_i \cdot \mu_i = \sum_{(i,p)\in\mathcal P} \tau_{i,p}(u) \cdot u_{i,p} + \sum_{j\in J} \lambda_j \cdot G_j(u)}$.
  Thus, the primal and dual objective values coincide for~$u$ and~$(\lambda, \mu)$, so strong duality implies that~$u$ is primal optimal.

  Conversely, assume that~$u$ is optimal for the primal problem, and let~$(\lambda, \mu)$ be an optimal solution to the dual problem.
  Let~$(i,p)\in\mathcal P$ be arbitrary with~$u_{i,p} > 0$.
  Then, the corresponding dual constraint must be tight, i.e.,~$\sum_{j\in J} (-g_{i,p,j}(u))\cdot \lambda_j + \mu_i = \tau_{i,p}(u)$.
  As~$d_i$ is positive, and~$(\lambda, \mu)$ optimal for the dual problem, we again have~$\mu_i = \min_{p\in \mathcal P_i} \tau_{i,p}(u) + \sum_{j\in J} \lambda_j \cdot g_{i,p,j}(u)$.
  This shows that~$c_{i,p}^\lambda(u) = \min_{q\in \mathcal P_i} c_{i,q}^\lambda(u)$, and thus~$u$ is a Wardrop equilibrium with respect to~$c^\lambda$.
\end{proof}
}
    \proofOfThmImplementabilityLPCharacterization

\subsection{Implementability of Budgets}

We continue with an analysis of implementability of \emph{budgets}.
As with flows, if arbitrary externality functions are allowed, not every feasible budget vector is implementable, as the following example illustrates.
Therefore, we first consider constant externalities, for which we show that any feasible budget is implementable.
Afterwards, we investigate a smoothness condition that yields sufficient implementability conditions for more general classes of externality functions.

\begin{example}\label{example:linear-externality-non-implementable-budget}
    Consider a network with a single externality class and two parallel edges~$e_1$ and~$e_2$, each with a linear externality function~$g_{e_i}(x)\coloneqq x_{e_i}$ for $i\in\{1,2\}$ such that the total externality is $G(x) = x_{e_1}^2 + x_{e_2}^2$.
    There is a single commodity with demand~$1$.
    Then, the budget~$B\coloneqq \nicefrac{1}{2}$ is feasible: There is a unique flow~$x^*$ respecting~$B$, namely the flow that assigns~$x^*_{e_i}\coloneqq \nicefrac{1}{2}$ to each edge.
    However, for travel time functions $\tau_{e_1}(x)\coloneqq 0$ and $\tau_{e_2}(x)\coloneqq 1$, the flow $x^*$ is never a Wardrop equilibrium for any choice of externality price $\lambda$:
    Since both edges are used, the equilibrium condition would imply $0+\lambda \cdot x^*_{e_1} = 1 + \lambda \cdot x^*_{e_2}$, a contradiction.
    Thus, neither~$x^*$ nor~$B$ is implementable.

    Note, however, that rearranging the equality for the Wardrop condition for a flow $x$ (and using $x_{e_1}+x_{e_2} = 1$) yields $x_{e_1} = \frac{1+\lambda}{2\lambda}$, and hence, as $\lambda$ tends to $\infty$, the equilibrium converges to $x^*$ and the total externality converges to $\nicefrac{1}{2}$.
    In particular, in this network, every strictly feasible budget is implementable.
\end{example}

\subsubsection*{Constant Externalities.}
If we, however, restrict ourselves to the case where all externality functions are constant, we can show that every feasible budget vector is implementable.
In the following, we call a feasible budget vector~$B$ \emph{Pareto-minimal} if there is no other feasible budget vector~$B'$ with~$B'_j \leq B_j$ for all~$j\in J$ and~$B' \neq B$.

\begin{theorem}\label{thm:all-feasible-budgets-implementable-kakutani}
  For continuous~$\tau$ and constant~$g$, every feasible budget vector~$B$ is implementable.
\end{theorem}
\begin{proof}
  It suffices to show that every Pareto-minimal budget vector is implementable.
  Therefore, let~$B$ be a Pareto-minimal budget vector and let~$\feasFlows^*$ denote the subset of feasible flows~$x\in\feasFlows$ respecting~$B$.
  Consider the set-valued function~$\corresp : \feasFlows^* \rightrightarrows \feasFlows^*$ given by \[
    \corresp(u) \coloneqq \argmin_{x\in\feasFlows^*} \sum_{(i,p)\in\mathcal P} \tau_{i,p}(u) \cdot x_{i,p}.
  \]
  Then, the set of fixed points~$u$ of~$\corresp$, i.e., the flows that fulfill~$u\in\corresp(u)$, coincides with the set of implementable Wardrop equilibria that respect~$B$:
  \begin{claim}\label{cl:fixed-point}
    A feasible flow~$u$ respecting~$B$ is implementable if and only if~$u$ is a fixed point of~$\corresp$.
  \end{claim}
  \begin{subproof}
    Every feasible flow~$x$ respecting~$B$ already fulfills~$G(x) = B$, as~$B$ is Pareto-minimal.
    Furthermore, since all~$g_{i,p,j}$ are constant, 
    the feasible set of the linear program \eqref{lp:characterization-implementability} for~$u$ is exactly~$\feasFlows^*$.
    Therefore,~$u$ is contained in~$\corresp(u)$ if and only if~$u$ minimizes \eqref{lp:characterization-implementability}.
    By \Cref{thm:implementability-lp-characterization}, this is equivalent to~$u$ being implementable.
  \end{subproof}

  We verify the conditions of Kakutani's theorem to establish that~$\corresp$ has a fixed point.
  \begin{description}
    \item[$\feasFlows^*$ is nonempty, compact and convex:] The set is a bounded polytope and nonempty by feasibility of~$B$.
    \item[$\corresp$ has nonempty and convex values:] In fact,~$\corresp(x)$ is the set of optimal solutions for a linear problem over the polytope~$\feasFlows^*$, and therefore has a convex optimal solution set.
    \item[$\corresp$ has a closed graph:] Let~$(u^n, x^n)$ be a sequence in~$\mathrm{graph}(\corresp)\coloneqq \{ (u, x) \mid u\in\feasFlows^*, x\in \corresp(u) \}$ that converges in~$\R^{\mathcal P}\times\R^{\mathcal P}$ to some~$(u^*, x^*)$.
    We aim to show~$(u^*, x^*)\in\mathrm{graph}(\corresp)$.
    As~$\feasFlows^*$ is closed, we have~$u^*, x^*\in\feasFlows^*$.
    To show that~$x^* \in \Gamma(u^*)$, we have to show that~$x^*$ minimizes the function~$T^{u^*}: x\mapsto \sum_{i,p}\tau_{i,p}(u^*) \cdot x_{i,p}$ over~$\feasFlows^*$.
    Thus, let~$y$ be any flow in~$\feasFlows^*$.
    Clearly, for every~$n\in\N$ we have \[
      \sum_{i,p} \tau_{i,p}(u^n) \cdot x^n_{i,p} \leq \sum_{i,p} \tau_{i,p}(u^n)\cdot y_{i,p}.
    \]
    As the left side converges to~$\sum_{i,p} \tau_{i,p}(u^*) \cdot x^*_{i,p}$ and as the right side converges to $\sum_{i,p} \tau_{i,p}(u^*) \cdot y_{i,p}$ due to the continuity of~$\tau$, we know that \[
      \sum_{i,p} \tau_{i,p}(u^*) \cdot x_{i,p}^* \leq \sum_{i,p} \tau_{i,p}(u^*) \cdot y_{i,p}
    \]
    which implies that~$x^*$ indeed minimizes~$T^{u^*}$.
  \end{description}
  Consequently, the Kakutani fixed-point theorem ensures the existence of a fixed point~$u \in \corresp(u)$.
  By \Cref{cl:fixed-point}, this fixed point corresponds to a feasible and implementable flow that respects budget~$B$, thereby establishing that~$B$ is indeed implementable.
\end{proof}

As the proof of \Cref{thm:all-feasible-budgets-implementable-kakutani} is based on a fixed-point argument, this result does not directly lead to a constructive way of finding the flow or the corresponding price vector that implements the given budget vector.
However, for the special class of so-called \emph{extremal} budget vectors, a different, more constructive approach for finding price vectors and corresponding Wardrop equilibria is possible.
Here, we say that a feasible budget vector~$B$ is \emph{extremal} if it cannot be written as a non-trivial convex combination of two distinct feasible budget vectors.
For these budget vectors, we can even show \emph{strict} implementability, meaning that every Wardrop equilibrium with respect to some fixed price vector respects the budget vector.

\begin{theorem}
\label{thm:pareto-implementable}
  For continuous~$\tau$ and constant~$g$, every extremal budget vector~$B$ is strictly implementable.
\end{theorem}
\begin{proof}
  The set of feasible budget vectors is a polyhedron in~$\smash{\R_{\geq0}^J}$ because all~$g_{i,p,j}$ are constant.
  Furthermore, every extremal budget vector is also Pareto-minimal.
  Thus, there is some~$\smash{\lambda\in\R^J_{\geq 0}}$ that strictly separates~$B$ from the remaining set of feasible budget vectors,
  i.e.,~$\lambda^T \cdot B < \lambda^T \cdot B'$ for any feasible budget vector~$B'\neq B$.

  For each commodity~$i\in I$, let~$h_i^*\coloneqq \min_{p\in\mathcal P_i} \lambda^\top g_{i,p}$.
  We then define~$\alpha$ to be any number larger than
  \begin{equation}\label{eq:strictly-implementable-alpha0}
    \alpha_0 \coloneqq
    \max_{i\in I}
    \frac{
      \max\{ \tau_{i,q}(x) - \tau_{i,p}(x) \mid p,q\in\mathcal P_i, x\in\feasFlows \}
    }{
      \min\{ \lambda^\top g_{i,p} - h_i^* \mid p\in \mathcal P_i, \lambda^\top g_{i,p} > h_i^* \}
    },
  \end{equation}
  where we define the fraction to be~$0$ for all commodities~$i$ with $\lambda^\top g_{i,p} = h_i^*$ for all $p \in \mathcal{P}_i$.

  Let~$x^*$ be any Wardrop equilibrium with respect to the cost function~$c^{\alpha\cdot \lambda}$.
  We show that~$x^*$ respects~$B$.
  
  Assume that there is a commodity~$i\in I$ and a path~$\smash{p\in\mathcal P_i}$ with~$\smash{x^*_{i,p} > 0}$ and~$\smash{\lambda^\top g_{i,p} > h_i^*}$.
  Choose a path~$q\in\mathcal P_i$ with~$\lambda^\top g_{i,q} = h_i^*$.
  Then, by the equilibrium condition we have~$\tau_{i,p}(x^*) + \alpha\cdot \lambda^\top g_{i,p} \leq \tau_{i,q}(x^*) + \alpha\cdot \lambda^\top g_{i,q}$, or equivalently~$\smash{\tau_{i,q}(x^*) - \tau_{i,p}(x^*) \geq \alpha\cdot \lambda^\top \cdot (g_{i,p} - g_{i,q})}$.
  However, the choice of~$\alpha$ yields the contradiction
  \[
  \alpha \cdot \lambda^\top \cdot (g_{i,p} - g_{i,q})
  > \alpha_0  \cdot \lambda^\top \cdot (g_{i,p} - g_{i,q})
  \geq \tau_{i,q}(x^*) - \tau_{i,p}(x^*).
  \]  

  Therefore, only paths that minimize~$p\mapsto \lambda^\top \cdot g_{i,p}$ over~$\mathcal P_i$ can have positive flow.
  This means, that~$x^*$ minimizes~$x\mapsto \lambda^\top G(x)$ over all feasible flows~$x$.
  In particular, we have~$\lambda^\top G(x^*) \leq \lambda^\top B$ which implies~$G(x^*) = B$ by the choice of~$\lambda$.
\end{proof}

Note that while non-extremal but Pareto-minimal budget vectors are implementable (as per \Cref{thm:all-feasible-budgets-implementable-kakutani}), the following example shows that they are not necessarily \emph{strictly} implementable.

\begin{example}
    Consider a network with two externality classes and two parallel edges $e_1,e_2$ with externality factors~$g_{e_1} = (1, 0)$ and~$g_{e_2} = (0, 1)$, respectively, and a zero travel time.
    For a single commodity with demand~$1$, the budget~$B=(\nicefrac{1}{2}, \nicefrac{1}{2})$ is Pareto-minimal.
    Let~$x$ denote the unique flow respecting~$B$ given by~$x_{e_1}=\nicefrac{1}{2}=x_{e_2}$.
    Then, for any price vector~$\lambda$ for which~$x$ is a Wardrop equilibrium, we must have~$\lambda_1 = \lambda_2$.
    However, any flow split between the two edges would result in an equilibrium with respect to~$\lambda$, and thus~$B$ is not strictly implementable.
\end{example}

As a consequence, in the single externality-class setting, all feasible budget vectors are strictly implementable.
We can explicitly describe a threshold such that for every price above this threshold, every associated War\-drop equilibrium has the minimum total externality.

\begin{restatable}{corollary}{CorImplementabilitySingleResources}\label{cor:implementability-single-resources}
  Assume~$|J| = 1$ and that the travel time functions~$\tau_{i,p}$ are continuous and all externality functions~$g_{i,p}$ are constant. 
  Then, every feasible budget is strictly implementable.
  
    More specifically, every $x\in\WE(\lambda)$ respects the minimal feasible budget for all~$\lambda > \max_{i\in I}\nicefrac{(\overline{\tau_i} - \underline{\tau_i})}{(g_i^\circ - g_i^*)}$, where~$\overline{\tau_i}\coloneqq \sup_{x\in \feasFlows, p\in \mathcal P_i} \tau_{i,p}(x)$, $\underline{\tau_i}\coloneqq \inf_{x\in \feasFlows, p\in\mathcal P_i} \tau_{i,p}(x)$, and where~$g_i^*\coloneqq \min_{p\in\mathcal P_i} g_{i,p}$ is the smallest externality on any of~$i$'s paths, and~$g_i^\circ\coloneqq \inf\{g_{i,p} \mid p\in\mathcal P_i, g_{i,p} > g_i^* \}$ is the second-smallest{} externality factor (or~$\infty$ if all paths in~$\mathcal P_i$ have the same externality).
\end{restatable}
\newcommand{\proofOfCorImplementabilitySingleResources}{
\begin{proof}
  We use the same bound as in the proof of the last theorem.
  As the set of feasible budgets is an interval, we may instantiate the separating vector in this proof with $\lambda\coloneqq 1$.
  Then, the value~$h_i^*$ coincides with~$g_i^*$ and the denominator of~$\alpha_0$ defined in~\eqref{eq:strictly-implementable-alpha0} is exactly $g_i^\circ - g_i^*$ (given $g_i^\circ < \infty$).
\end{proof}
}
\proofOfCorImplementabilitySingleResources

\subsubsection*{General Externalities.}

In the following, we elaborate on the following question for general externality functions:
Given a feasible budget~$B$, by how much do we need to increase the budget in order to be able to implement it?
Clearly, \Cref{example:linear-externality-non-implementable-budget} already shows that not all feasible budgets are implementable; however, in this example, any budget strictly higher than the minimum feasible budget is implementable.
We will obtain this result for any edge-separable externality functions that are linear in the edge load (in fact, even for more general functions).

As we will show, there is an interesting parallel to the price of anarchy for Wardrop equilibria.
Thus, let us first reiterate some helpful literature on the latter.
The \emph{price of anarchy} (PoA) is, for a given set of problem instances (typically characterized by a class of travel time functions), the worst case ratio of the total travel time $\scalprod{\tau(x)}{x}$ of a Wardrop equilibrium flow and a system optimum.

The PoA has been studied extensively starting with the work by \citet{Roughgarden02}, who proved that the PoA is $4/3$ for separable and affine travel time functions.
Analyses for general classes of travel time functions were subsequently conducted by \citet{Roughgarden2003} and \citet{Correa2004,CorreaSM08}.
The latter showed that the PoA of a function class $\mathcal{C}$ consisting of continuous, monotone\footnote{\label{fn:monotone}Here, $c: \R^d_{\geq0} \to \R^d_{\geq0}$ is \emph{monotone} if $\scalprod{c(x) - c(y)}{x - y}\geq 0$ holds for all $x,y\in\R_{\geq0}^d$.} functions $\tau:\R^{\mathcal{P}}_{\geq0}\to\R^{\mathcal{P}}_{\geq0}$ is bounded from above by $\frac{1}{1-\beta(\mathcal C)}$ (with $1/0\coloneqq\infty$) where
\[
  \beta(\mathcal{C}) \coloneqq
  \sup_{\tau\in\mathcal{C}, x\in\feasFlows}
  \max_{y\in\R_{\geq0}^{\mathcal{P}}}
  \frac{\scalprod{\tau(x) - \tau(y)}{y}}{\scalprod{\tau(x)}{x}}
  \qquad
  \text{with}
  \quad
  \frac{0}{0}\coloneqq0,
\]
and this bound is tight for separable cost functions assuming $\mathcal{C}$ contains all constant functions.
As laid out by \citet[Section~6]{Roughgarden2015a}, a tight bound even in the absence of constant functions requires the use of two parameters, which leads us to their definition of $(\alpha,\beta)$-smoothness\footnote{\citet{Roughgarden2015a} calls this property $(\lambda,\mu)$-smoothness.}:

\begin{definition}
  For given $\alpha\in\R_{>0}$ and $\beta\in[0,1)$, a function $c: \feasFlows \to \R_{\geq0}^{\mathcal{P}}$ is called 
  \emph{$(\alpha,\beta)$-smooth} if it fulfils 
  \begin{equation}\label{eq:alpha-beta-smoothness}
    \scalprod{c(x)}{y} \leq \alpha \cdot \scalprod{c(y)}{y} + \beta \cdot \scalprod{c(x)}{x}
    \qquad
    \text{for all $x,y\in\feasFlows$}.
  \end{equation}
  For a given class $\mathcal{C}$ of functions $c:\feasFlows\to\R_{\geq0}^{\mathcal{P}}$, define
  \[
    R(\mathcal{C}) \coloneqq \sup_{c\in\mathcal{C}} \left(\inf\left\{\, \frac{\alpha}{1-\beta} \,\middle|\, \alpha\in\R_{>0}, \beta\in[0,1), \text{ $c$ is $(\alpha,\beta)$-smooth } \right\}\right).
  \]
\end{definition}

For an $(\alpha,\beta)$-smooth travel time function $\tau$,
rearranging~\eqref{eq:alpha-beta-smoothness} instantiated with a Wardrop equilibrium $x$ and a system optimum $y$ yields
\[
  \alpha\cdot \scalprod{\tau(y)}{y} \geq \scalprod{\tau(x)}{y - \beta\cdot x}
  = \scalprod{\tau(x)}{y-x} + (1 - \beta)\cdot \scalprod{\tau(x)}{x}
  \geq (1-\beta)\cdot \scalprod{\tau(x)}{x},
\]
and hence the ratio of the total travel times is at most $\frac{\alpha}{1-\beta}$.
Therefore, $R(\mathcal{C})$ is an upper bound on the PoA.
Note that all $c\in\mathcal{C}$ are $(1,\beta(\mathcal{C}))$-smooth provided $\beta(\mathcal{C}) < 1$, in which case $R(\mathcal{C})\leq \frac{1}{1-\beta(\mathcal{C})}$.

We now return to our original question on budget implementability.
If we assume that our \emph{externality functions} fulfill such a smoothness condition, we may prove that scaling any strictly feasible budget $B$ by the factor $\frac{\alpha}{1 - \beta}$ makes the budget implementable.

\begin{theorem}\label{thm:smoothness-implementability}
  Assume $\tau$ and $g$ are continuous, and for each $j\in J$, there are $\alpha_j,\beta_j$ such that $g_j$ is $(\alpha_j,\beta_j)$-smooth.
  Then, for any strictly feasible budget $B$, the budget
  $B'\coloneqq (  \frac{\alpha_j}{1-\beta_j}  \cdot B_j  )_{j\in J}$
  is implementable.

  More specifically, if $y$ is a flow strictly respecting $B$, there exist prices $\lambda\in\R_{\geq0}^J$ implementing $B'$ with
  \begin{equation}\label{eq:choice-delta-and-D}
    \lambda_j \leq \frac{D}{\alpha_j \cdot \delta_j} \qquad \text{where}
    \quad
    D\coloneqq\max_{x\in\feasFlows} \scalprod{\tau(x)}{y - x},
    \quad
    \text{and}
    \quad
    \delta_j\coloneqq B_j-G_j(y) >0.
  \end{equation}
\end{theorem}
\begin{proof}
  As $\tau$ is continuous, the definition of $D$ is well-defined.
  Choose any $M_j>D/(\alpha_j\cdot\delta_j)$ and define $\Lambda\coloneqq\bigtimes_{j\in J}[0,M_j]$.
  Consider the set-valued mapping
  \begin{align*}
    &\corresp:
    \feasFlows\times\Lambda
    \rightrightarrows
    \feasFlows\times\Lambda,
    \\
    &(x, \lambda) \mapsto
    \Bigg( \argmin_{z\in\feasFlows}\,
       \scalprod{\tau(x) + \sum_{j\in J} \lambda_j \cdot g_j(x)}{z} \Bigg)
    \,
    \times
    \,
    \Bigg( \argmax_{\mu\in\Lambda}\,
       \scalprod{\mu}{G(x) - B'} \Bigg)
  \end{align*}
  We apply Kakutani's fixed-point theorem to $\corresp$.
  Its domain is nonempty, compact, and convex.
  Both optimization problems have nonempty compact solution sets as their objective functions are continuous, and both solution sets are convex.
  Finally, $\corresp$ has a closed graph, since the objective functions of the problems are continuous in the input.
  Hence, Kakutani's theorem yields a fixed point
  $(x^*,\lambda^*)\in\corresp(x^*,\lambda^*)$.

  By the first component of the fixed-point condition, $x^*$ solves the variational inequality~\eqref{eq:variational-ineq} for $c^{\lambda^*}$.
  Therefore, $x^*$ is a Wardrop equilibrium for the price vector~$\lambda^*$.
  We conclude by showing that $x^*$ respects~$B'$.

  First, note that the optimality of the second component implies
  \begin{equation}\label{eq:complementarity-condition}
    \lambda^*_j \cdot (G_j(x^*) - B_j') \geq 0
  \end{equation}
  holds for all $j\in J$, and $\lambda^*_j = M_j$ holds if $G_j(x^*) > B_j'$.
  The optimality of the first component implies
  \[
    \scalprod{\tau(x^*)}{y - x^*} \geq \sum_{j\in J}\lambda_j^* \cdot \scalprod{g_j(x^*)}{x^* - y}.
  \]
  We next bound the right side from below.
  Using \eqref{eq:alpha-beta-smoothness}, we obtain
  \[
    \begin{aligned}
      \scalprod{g_j(x^*)}{x^*-y}
      &=G_j(x^*)-\scalprod{g_j(x^*)}{y}\\
      &\geq(1-\beta_j) \cdot G_j(x^*)- \alpha_j \cdot G_j(y).
    \end{aligned}
  \]
  Multiplying by $\lambda_j^*$ and summing over $j$ therefore gives (together with the definition of $D$ in~\eqref{eq:choice-delta-and-D})
  \[
    D \geq \scalprod{\tau(x^*)}{y-x^*}\geq \sum_{j\in J}\lambda_j^*\cdot
      \bigl((1-\beta_j)\cdot G_j(x^*) - \alpha_j \cdot G_j(y)\bigr).
  \]
  Using $(1-\beta_j)\cdot B'_j= \alpha_j \cdot B_j$, the right side can be written as
  \[
    \sum_{j\in J}\lambda_j^* \cdot (1-\beta_j)\cdot
      \bigl(G_j(x^*)-B'_j\bigr)
    +\sum_{j\in J}\lambda_j^* \cdot \alpha_j\cdot \bigl(B_j-G_j(y)\bigr).
  \]
  The first sum is nonnegative by \eqref{eq:complementarity-condition} and we obtain $D\geq \sum_{j\in J} \lambda_j^*\cdot \alpha_j\cdot \delta_j$.
  As all summands are nonnegative, this implies $\lambda_j^* \leq \frac{D}{\alpha_j\cdot \delta_j} < M_j$ for all $j\in J$.
  Violating a budget constraint $j\in J$ would imply $\lambda_j^*=M_j$, a contradiction.
  Hence, $G_j(x^*)\leq B'_j$ for every $j\in J$, and $\lambda^*$ implements $B'$.
\end{proof}

We continue with analogous statements on the implementability bound as for the price of anarchy.
For a class $\mathcal{C}$ of continuous functions  $c:\feasFlows\to\R_{\geq0}^{\mathcal{P}}$, define the \emph{implementability factor} $\sigma(\mathcal C)$ as the infimum among all $\alpha\geq0$ such that, in every instance with continuous travel time functions and with externality functions drawn from $\mathcal{C}$, scaling any strictly feasible budget by $\alpha$ makes it implementable.
Here, $\mathcal{C}$ is again a class of continuous functions $c:\feasFlows\to\R_{\geq0}^{\mathcal{P}}$.
Then, the previous theorem shows:

\begin{corollary}\label{cor:upper-bound-of-implementability-factor}
  For a class $\mathcal C$ of continuous functions $c:\feasFlows\to\R_{\geq0}^{\mathcal{P}}$, the implementability factor $\sigma(\mathcal{C})$ is bounded from above by $R(\mathcal{C})$.
\end{corollary}

In particular, the implementability factor $\sigma(\mathcal{C})$ inherits most upper bounds on the PoA proven in the literature as they implicitly or explicitly are based on the same smoothness argument.
To exemplify this, \citet[Theorem~2.5]{Per07} implicitly\footnote{
  After carefully analyzing their proof, it becomes clear that the arguments therein show, in our notation and for arbitrary $x,y\in\R_{\geq0}^{\mathcal{P}}$, that $c(x)^\top y \leq (1 + A(a_2 - 1)) \cdot c(y)^\top y + A a_1 c(x)^\top x$ for any $a_1\in[0,1/A]$ and $a_2 \geq 1$ with $a_1a_2\geq\gamma^2/ 4$.
  Suitable choices of $a_1$ and $a_2$, under a case distinction on $\gamma^2\leq\nicefrac{2}{A}$, yield the presented bounds.
} proved the following result to obtain their PoA bound for nonseparable, nonlinear and asymmetric travel time functions:
\begin{lemma}
  Let $c: \R_{\geq0}^{\mathcal P}\to \R_{\geq0}^{\mathcal{P}}$ be continuously differentiable with Jacobian $\Jac c$ and assume its symmetric part $S(x)\coloneqq \frac{1}{2}((\Jac c)(x) + (\Jac c)(x)^\top)$ is positive definite for every $x\in\R_{\geq0}^{\mathcal P}$.
  Suppose there are constants $A\geq 1$ (measuring nonlinearity) and $\gamma\geq 1$ (measuring asymmetry) with
  \[
     \forall x,y\in\R_{\geq0}^{\mathcal{P}}, v\in\R^{\mathcal{P}}:
     \quad
     v^\top S(x) v \leq A\cdot v^\top S(y) v
     \quad\text{and}\quad
     \norm{S(x)^{-\nicefrac{1}{2}} (\Jac c)(x) S(x)^{-\nicefrac{1}{2}}}_2 \leq \gamma.
  \]
  Then, $c$ is $(\alpha,\beta)$-smooth with
  \[
    (\alpha,\beta) = \begin{cases}
      (1, \frac{\gamma^2 A}{4}), &\text{if $\gamma^2 \leq \frac{2}{A}$},\\
      (\frac{\gamma^2 A^2}{2} - A + 1, \frac{1}{2}), &\text{otherwise},
    \end{cases}
    \quad\text{and}\quad
    \frac{\alpha}{1- \beta} = \begin{cases}
      \frac{4}{4- \gamma^2 A}, &\text{if $\gamma^2 \leq \frac{2}{A}$,}\\
      \gamma^2 A^2 - 2(A - 1), &\text{otherwise.}
    \end{cases}
  \]
\end{lemma}

Regarding the tightness of the bound in \Cref{cor:upper-bound-of-implementability-factor}, note the following proposition.
For this, the \emph{price of stability} (PoS) of a function class $\mathcal{C}$ is the supremum, among all instances with travel time function drawn from $\mathcal{C}$, of the total travel time ratio of the \emph{best} Wardrop equilibrium and the system optimum.

\begin{proposition}
  For a class $\mathcal C$ of continuous functions $c:\feasFlows\to\R_{\geq0}^{\mathcal{P}}$, the implementability factor $\sigma(\mathcal{C})$ is bounded from below by the PoS of $\mathcal{C}$.
\end{proposition}
\begin{proof}
  For this, consider, for any $\varepsilon > 0$, an instance together with a corresponding travel time function $\tau\in\mathcal{C}$ and a corresponding equilibrium $x^*$ of minimum total travel time, such that the ratio $r\coloneqq \scalprod{\tau(x^*)}{x^*}/\min_{x\in\feasFlows}\scalprod{\tau(x)}{x}$ is at most $\varepsilon$ lower than the PoS, if the PoS is finite, and is at least $\nicefrac{1}{\varepsilon}$, if the PoS is unbounded.
  We now augment this instance with a single externality $|J|=1$ whose externality function is given by $g \coloneqq \tau$.
  Then, the set of Wardrop equilibria is independent of the price $\lambda\geq0$ (and thus coincides with the equilibria without externality), since the perceived cost function $\tau + \lambda\cdot g$ is simply $\tau$ scaled by a factor of $(1+\lambda)$.
  Further, the minimum feasible budget $B$ coincides with the minimum total travel time.
  Now, for any $\delta > 0$, $B+\delta$ is strictly feasible while the minimum implementable budget is the total travel time of $x^*$.
  Therefore, we must have $\sigma(\mathcal{C}) \geq \frac{\scalprod{\tau(x^*)}{x^*}}{\min_{x\in\feasFlows}\scalprod{\tau(x)}{x}+\delta}$.
  As $\delta>0$ was arbitrary, this shows $\sigma(\mathcal{C})\geq \text{PoS} - \varepsilon$, if the PoS is finite, and $\sigma(\mathcal{C})\geq \nicefrac{1}{\varepsilon}$, if the PoS is unbounded.
  Because $\varepsilon >0$ was arbitrary as well, this proves that the price of stability is a lower bound on the implementability factor.
\end{proof}

In conclusion, we proved the following inequalities for every class $\mathcal{C}$ of continuous functions:
\begin{center}
  \begin{tikzpicture}[every node/.style={inner sep=3pt}]
    \node[minimum width=1.6cm] (pos) at (0,0) {$\mathrm{PoS}(\mathcal{C})$};
    \node[minimum width=1.6cm] (poa) at (2.1,0.5) {$\mathrm{PoA}(\mathcal{C})$};
    \node[minimum width=1.6cm] (sigma) at (2.1,-0.5) {$\sigma(\mathcal{C})$};
    \node[minimum width=1.6cm] (smoothness) at (4.2,0) {$R(\mathcal{C})$};
    \path (pos.north east) -- node[midway,sloped,anchor=center] {$\leq$} (poa.west);
    \path (poa.east) -- node[midway,sloped,anchor=center] {$\leq$} (smoothness.north west);
    \path (pos.south east) -- node[midway,sloped,anchor=center] {$\leq$} (sigma.west);
    \path (sigma.east) -- node[midway,sloped,anchor=center] {$\leq$} (smoothness.south west);
  \end{tikzpicture}
\end{center}
Note that for many practically relevant classes of continuous functions, the total travel time is unique among all equilibria of a given instance, and thus, the PoA and the PoS of such classes collapse.
Furthermore, if such a class has the additional property that the smoothness bound $\text{PoA}(\mathcal{C})\leq R(\mathcal{C})$ is tight, the PoA of $\mathcal{C}$ coincides with the implementability factor $\sigma(\mathcal{C})$.
Examples of such classes are classes induced, by assigning to each edge independently and over arbitrary networks, a function from a family of continuous and nondecreasing functions $c_e : \R_{\geq0}\to\R_{\geq0}$ containing all nonnegative constant functions, see \citep{CorreaSM08}.
More examples of classes with tight bounds are presented in \citep{Roughgarden2003,Per07}.

\subsection{PPAD-Hardness}

\newcommand{\BIMATRIX}{\textsc{Bimatrix}}
\newcommand{\PPAD}{PPAD}

The previous results used fixed point arguments to obtain existence of prices and a corresponding equilibrium that respects a certain budget vector.
Unfortunately, this means that the results are of little practical help when a price and an equilibrium are to be determined in a real-world scenario.
In fact, for the case without externalities (and budgets), \citet{Klimm2025} proved that computing a Wardrop equilibrium even for the case of affine, edge-separable but commodity-specific travel time functions, is PPAD-hard.
Thus, the problem of computing prices and a corresponding equilibrium respecting a given budget is PPAD-hard, too.
In the following, we give an alternative (and much simpler) proof of PPAD-hardness for the case of linear, nonseparable travel time functions that, compared to \citep{Klimm2025}, works even for trivial graphs consisting of parallel edges.

\begin{lemma}\label{lem:ppad-hardness}
  The problem of finding a Wardrop equilibrium is PPAD-hard, even if the travel times are linear, the graph consists only of parallel edges between two nodes, and there are two commodities.
\end{lemma}
\begin{proof}
  We use the formulation of the \BIMATRIX{} problem as defined by~\citet{Chen09complexity} with the modification that we consider cost minimization rather than profit maximization:
  An instance of \BIMATRIX{} is specified by two rational matrices $A,B\in [0, 1]^{n\times n}$ where $n\in\N$ is the number of pure strategies for each of the two players.
  A \emph{mixed strategy} is a vector $x\in[0,1]^n$ satisfying $\sum_{k\in[n]} x_k = 1$.
  The cost of a \emph{mixed strategy profile}, i.e., of a pair $(x, y)$ of mixed strategies, is given by $x^\top A \,y$ and $x^\top B \, y$ for player $1$ and $2$, respectively.
  The search problem consists of finding a \emph{mixed Nash equilibrium}, i.e., a mixed strategy profile $(x^*, y^*)$ that fulfils
  \begin{equation}\label{eq:mixed-nash-equilibrium}
    (x^*)^\top A \, y^* \leq x^\top A \, y^*
    \quad \text{ and } \quad
    (x^*)^\top B \, y^* \leq (x^*)^\top B \, y
  \end{equation}
  for all mixed strategies $x$ and $y$.
  In other words, \BIMATRIX{} asks to find a mixed strategy profile in which no player can strictly decrease their cost by unilaterally changing their mixed strategy.
  
  Note that the complexity of \BIMATRIX{} does not depend on whether the players maximize payoffs or minimize costs, where for the former the inequalities in~\eqref{eq:mixed-nash-equilibrium} are swapped:
  By taking $\hat{A} \coloneqq 1_{n\times n} - A$ and $\hat{B} \coloneqq 1_{n\times n} - B$ where $1_{n\times n}$ is the matrix with all entries equal to $1$, we can translate a profit maximization instance into a cost minimization instance and vice versa, without changing the set of mixed Nash equilibria.
  The resulting matrices $\hat A$ and $\hat B$ are still rational matrices in $[0,1]^{n\times n}$.
  \citet{Chen09complexity} showed that \BIMATRIX{} is \PPAD-complete.

  Given a (cost minimization) instance of \BIMATRIX, we now construct a network with two commodities~$I\coloneqq \{1, 2\}$ and common source and sink nodes~$s$ and $t$.
  For each pure strategy $k\in [n]$, we introduce an edge from~$s_i$ to~$t_i$, which we denote as~$e_{k}$.
  As we set the demand to $d_1 = d_2 = 1$, the set of demand-feasible flows corresponds one-to-one to the set of mixed strategies, and we denote such a flow as a pair $(x, y)$.
  The travel times on the edges are given by
  \[
    \tau_{1,k}(x, y) \coloneqq \sum_{l\in [n]} A_{k,l} \cdot y_{l}
    \quad \text{ and }\quad
    \tau_{2,l}(x, y) \coloneqq \sum_{k\in [n]} x_{k} \cdot B_{k,l}
  \]
  for $k,l\in[n]$.
  We now show that a flow $(x^*, y^*)$ is a Wardrop equilibrium in this network if and only if it is a mixed Nash equilibrium for the cost minimization instance of \BIMATRIX:
  First, let us rewrite the equivalent condition of Wardrop equilibria as per \Cref{lem:wardrop-eq-vi}:
  \[
    \scalprod{\tau(x^*, y^*)}{(x - x^*, y - y^*)} \geq 0 \quad \text{ for all } \quad x,y \in \Bigl\{ x'\in[0,1]^n \;\Big\vert\; \textstyle\sum_{k\in[n]} x'_k = 1 \Bigr\}
  \]
  where the left-hand side of the inequality equates to
  \[
    \sum_{k\in[n]} \Bigg(\sum_{l\in[n]} A_{k,l} y^*_l \Bigg) \cdot (x_k - x^*_k)
    + \sum_{l\in[n]} \Bigg(\sum_{k\in[n]} x^*_k B_{k,l} \Bigg) \cdot (y_l - y^*_l)
    = (x - x^*)^\top A \, y^* + (x^*)^\top B \, (y - y^*).
  \]
  Instantiating this condition once with~$(x,y)\coloneqq (x,y^*)$ and once with $(x,y) \coloneqq (x^*,y)$ implies the mixed Nash equilibrium condition~\eqref{eq:mixed-nash-equilibrium} for player $1$ and $2$, respectively, for arbitrary mixed strategies $x$ and $y$.
  Conversely, if $(x^*, y^*)$ is a mixed Nash equilibrium, both summands in the above expression are nonnegative, and thus the Wardrop equilibrium condition is fulfilled.

  This transforms an instance of \BIMATRIX{} into a Wardrop equilibrium search problem of polynomial size, and since \BIMATRIX{} is \PPAD-hard, the latter problem is \PPAD-hard as well.
\end{proof}

\section{Convex potential functions}
\label{sec:potential}

The existence results in \Cref{sec:general} leave open how to find implementing prices and corresponding equilibria for given budgets.
In this section, we address this question when the travel time and externality functions admit convex potentials.
We first formulate a constrained convex program whose Lagrange multipliers yield implementing prices, and examine when this approach leads to implementability.
We then give an explicit price construction for smooth externalities that strictly implement relaxed feasible budgets.
Finally, we establish the monotonicity of total externality in the price used by the algorithms in \Cref{sec:affine}.

\begin{definition}
  We say that the travel time functions have a \emph{potential} if there exists a differentiable function~$\Phi: \R_{\geq0}^\mathcal P \to \R$ such that for all~$(i,p)\in \mathcal P$ and all~$x\in\feasFlows$ it holds that \[
    \tau_{i,p}(x) = \frac{\diff \Phi}{\diff x_{i,p}}(x) \quad \text{or equivalently}  \quad \tau(x) = \nabla \Phi (x).
  \]
\end{definition}

A simple case where the travel time functions have a potential is when the functions are \emph{separable}.
This means, that there exist functions~$\tau_e \colon \mathbb{R}_{\geq 0} \to \mathbb{R}_{\geq 0}$ for~$e\in E$ such that~$\tau_{i,p}(x)=\sum_{e\in p} \tau_e(x_e)$ holds for all~$(i,p)\in\mathcal P$ where~$x_e$ is defined as the total flow on edge~$e$, that is~$x_e\coloneqq \sum_{(i,p)\in\mathcal P} \mathbf{1}_{e\in p} \cdot x_{i,p}$.
These functions admit the well-known \emph{Beckmann--McGuire--Winsten (BMW) potential} defined as~$\smash{\Phi(x) \coloneqq \sum_{e\in E} \int_0^{x_e} \tau_e(y) \diff y}$.
If we additionally assume that each~$\tau_e$ is nondecreasing, then this potential is a convex function.
For convex potential functions, we know that a feasible flow is a Wardrop equilibrium if and only if it minimizes~$\Phi$ over~$\feasFlows$; see \citet{Beckmann56}.

\newcommand{\PropWardropIffMinimizesPotential}{
\begin{proposition}[\citet{Beckmann56}]\label{prop:wardrop-iff-minimizes-potential}
  Let~$c$ be a cost function with a convex potential function~$\Phi$.
  Then, a flow~$x\in\feasFlows$ is a Wardrop equilibrium with respect to~$c$ if and only if~$x$ minimizes~$\Phi$ over~$\feasFlows$.
  In particular, there exists a Wardrop equilibrium with respect to~$c$, and the set of Wardrop equilibria is convex.
\end{proposition}
}

\PropWardropIffMinimizesPotential

We say that the externality functions~$g_{i,p,j}$ have a \emph{potential} if there exist differentiable functions~$\smash{\Psi_j: \R_{\geq0}^{\mathcal{P}} \to \R}$ for~$j\in J$, such that for all~$(i,p)\in\mathcal P$ and all~$x\in\feasFlows$ it holds that~$g_{i,p,j}(x)=(\nabla\Psi_{j}(x))_{i,p}$.
If both~$\tau$ and~$g_{i,p,j}$ admit convex potentials~$\Phi$ and~$\Psi_j$, respectively, then for any price vector~$\smash{\lambda\in\R_{\geq0}^J}$, the cost function~$\smash{c^\lambda}$ admits the convex potential~$\smash{\Phi^\lambda(x) \coloneqq \Phi(x) + \sum_{j\in J} \lambda_j \cdot \Psi_j(x)}$.

\subsection{Implementability through constrained convex optimization}

The potential representation allows us to obtain an implementable equilibrium respecting a given budget by minimizing the travel time potential subject to suitable externality constraints.

\begin{theorem}\label{thm:implementability-for-potential-travel-time}
  Assume{} that both~$\tau$ and~$g$ have convex potential functions~$\Phi$ and~$(\Psi_j)_j$.
  Let~$\alpha_j \in\R_{>0}$ and $\beta_j\in\R$ be such that~$G_j(x) \leq \alpha_j \cdot \Psi_j(x) + \beta_j$ holds for all~$x\in\feasFlows$.
  For a budget vector~$B$, assume that the following convex optimization problem is feasible:
  \begin{equation}\label{eq:beckmann}\tag{P}
    \begin{aligned}
      \min_{x\in\mathcal F}~& \Phi(x)\\
      \text{s.t.}          ~& \alpha_j \cdot \Psi_j(x) + \beta_j \leq B_j \quad \text{for all } j\in J.
    \end{aligned}
  \end{equation}
  We call a vector $\lambda\in\R^J_{\geq0}$ a \emph{KKT (Karush--Kuhn--Tucker) multiplier vector} of~\eqref{eq:beckmann} if the value $\inf_{x \in \mathcal F} \Phi(x) + \sum_j \lambda_j \cdot (\alpha_j\cdot \Psi_j(x) + \beta_j - B_j)$ coincides with the optimal value of \eqref{eq:beckmann}.
  \begin{enumerate}
    \renewcommand{\theenumi}{(\roman{enumi})}
    \renewcommand{\labelenumi}{\textup{\theenumi}}
    \item\label{thm:implementability-for-potential-travel-time:kkt} For every KKT multiplier vector $\lambda$ of~\eqref{eq:beckmann}, the prices $(\alpha_j\cdot \lambda_j)_{j\in J}$ implement every optimal solution of~\eqref{eq:beckmann}, and thus implement~$B$.
    \item\label{thm:implementability-for-potential-travel-time:kkt-existence-iff-lp} Let $x^*$ be feasible for~\eqref{eq:beckmann}.
    Then, $x^*$ is optimal for~\eqref{eq:beckmann} and a KKT multiplier vector of~\eqref{eq:beckmann} exists if and only if $x^*$ minimizes
    \begin{equation}\label{eq:lp-characterization-kkt}
      \begin{aligned}
        \min_{x\in\feasFlows} ~&\scalprod{\tau(x^*)}{x}\\
        \text{s.t.} ~&\scalprod{g_j(x^*)}{x - x^*} \leq 0 &\text{for $j\in J$ with $\alpha_j \cdot \Psi_j(x^*) + \beta_j = B_j$.}
      \end{aligned}
    \end{equation}

    \item\label{thm:implementability-for-potential-travel-time:slater} A KKT multiplier vector of~\eqref{eq:beckmann} exists if there is a flow $x$ feasible for~\eqref{eq:beckmann} with $\alpha_j \cdot \Psi_j(x) + \beta_j < B_j$ for all $j\in J$ for which $\Psi_j$ is not affine.
  \end{enumerate}
\end{theorem}
\begin{proof}
  Let us denote the Lagrangian of~\eqref{eq:beckmann} as $L(\lambda, x)\coloneqq \Phi(x) + \sum_{j\in J} \lambda_j \cdot (\alpha_j \cdot \Psi_j(x) + \beta_j - B_j)$.

  \emph{Statement~\ref{thm:implementability-for-potential-travel-time:kkt}.}
  It follows from standard results on convex optimization (see \citealt[Theorem~28.1]{Rockafellar2015}) that for a given KKT multiplier vector $\lambda$, among the minimizers of $L(\lambda, x)$ over $x\in \feasFlows$, the vectors $x$ that are feasible for~\eqref{eq:beckmann} and fulfil $\alpha_j\cdot \Psi_j(x) + \beta_j = B_j$ for all $j$ with $\lambda_j > 0$ are exactly the optimal solutions of~\eqref{eq:beckmann}.
  Thus, an optimal solution $x^*$ of~\eqref{eq:beckmann} minimizes the convex potential $x\mapsto L(\lambda, x)$ over $x\in\feasFlows$, which is by \Cref{prop:wardrop-iff-minimizes-potential} equivalent to being a Wardrop equilibrium with respect to prices $(\alpha_j\cdot \lambda_j)_{j\in J}$.
  Finally, $x^*$ respects $B$, since we have $G_j(x^*) \leq \alpha_j\cdot \Psi_j(x^*) + \beta_j \leq B_j$ for all $j\in J$.
  As an optimal solution $x^*$ exists (as we assumed feasibility of \eqref{eq:beckmann}), $B$ is implementable.

  \emph{Statement~\ref{thm:implementability-for-potential-travel-time:kkt-existence-iff-lp}}.
  Let $x^*$ be feasible for~\eqref{eq:beckmann} and denote by $J^*\coloneq \{ j\in J \mid \alpha_j\cdot\Psi_j(x^*) + \beta_j = B_j \}$ the set of active constraints of $x^*$.
  If $x^*$ is even optimal for~\eqref{eq:beckmann} and $\lambda^*$ is a KKT multiplier vector of~\eqref{eq:beckmann}, then $\inf_{x\in \feasFlows} L(\lambda^*, x) = \Phi(x^*)$ implies $L(\lambda^*, x^*) \geq \Phi(x^*)$ and by definition of the Lagrangian, we have $\Phi(x^*) \geq L(\lambda^*, x^*)$, and thus equality.
  This also implies $\lambda_j^*=0$ for all $j\notin J^*$ by complementarity.
  Therefore, $x^*$ minimizes the function $L(\lambda^*, \,\cdot\,)$ over $\feasFlows$ with value $\Phi(x^*)$.
  Since this function is convex, we have $\scalprod{\nabla_x L(\lambda^*, x^*)}{x - x^*} \geq 0$ for all $x \in \feasFlows$.
  Expanding the definition of the Lagrangian, and using $\nabla \Phi(x^*) = \tau(x^*)$ and $\nabla\Psi_j(x^*) = g_j(x^*)$, we get for all $x\in\feasFlows$ that
  \[
    \scalprod{\tau(x^*)}{x - x^*} \geq - \sum_{j\in J}\lambda_j^*\scalprod{\alpha_j\cdot g_j(x^*)}{x - x^*} = -\sum_{j\in J^*} \lambda_j^*\cdot\scalprod{ \alpha_j \cdot g_j(x^*)}{x - x^*}.
  \]
  For flows $x$ feasible for~\eqref{eq:lp-characterization-kkt}, we additionally have $\scalprod{g_j(x^*)}{x - x^*} \leq 0$ for $j\in J^*$, and hence, $\scalprod{\tau(x^*)}{x-x^*} \geq 0$; this means $x^*$ minimizes the linear program.

  For the converse direction, we exploit that the ground set is a polyhedron of the form $\feasFlows = \{ x\in \R^{\mathcal P} \mid Ax \leq b \}$ for appropriately chosen $A$ and $b$.
  Let $x^*$ solve the linear program~\eqref{eq:lp-characterization-kkt}.
  Treating the constraint $x\in\feasFlows$ explicitly as $Ax \leq b$ and multiplying each active budget constraint by $\alpha_j>0$, strong duality of the linear program yields vectors $\lambda^*,\mu\geq0$ that are dual feasible, i.e., $\sum_{j\in J^*} \lambda^*_j\cdot \alpha_j\cdot g_j(x^*) + \mu^\top A = - \tau(x^*)$, and that fulfil complementary slackness with $x^*$, i.e., $\mu^\top(b - Ax^*) = 0$.
  We expand $\lambda^*$ by $\lambda^*_j \coloneq 0$ for $j\in J\setminus J^*$ and show that this yields a KKT multiplier vector:
  We consider the Lagrangian and use dual feasibility to get
  \[
    \nabla_x L(\lambda^*, x^*) = \tau(x^*) + \sum_{j\in J}\lambda^*_j\cdot \alpha_j \cdot g_j(x^*) = -\mu^\top A
  \]
  and thus, for any $x\in\feasFlows$, using convexity, we obtain $L(\lambda^*, x)\geq L(\lambda^*,x^*) + \scalprod{-\mu^\top A}{x - x^*} \geq L(\lambda^*, x^*)$ where in the last inequality, we use $\mu^\top (b - Ax^*) = 0$.
  Since $\lambda_j^*=0$ for $j\notin J^*$, we have $L(\lambda^*,x^*)=\Phi(x^*)$, so $\inf_{x\in\feasFlows}L(\lambda^*,x)=\Phi(x^*)$ as required.
  By weak duality, $x^*$ is optimal and $\lambda^*$ is a KKT multiplier for~\eqref{eq:beckmann}.

  \emph{Statement~\ref{thm:implementability-for-potential-travel-time:slater}.}
  This is the translation of the Slater constraint qualification to the problem~\eqref{eq:beckmann}:
  To apply \cite[Theorem~28.2]{Rockafellar2015}, we have to provide a flow in the relative interior of the feasibility set of~\eqref{eq:beckmann} that fulfills all non-affine inequality constraints strictly.
  To achieve this, mix the given flow satisfying the non-affine inequalities strictly with a flow in the relative interior of the feasible set of~\eqref{eq:beckmann}.
\end{proof}

Thus, whenever the constraints are feasible and KKT multipliers exist, convex optimization provides both a flow respecting the budgets and prices implementing that flow.
We illustrate the implications of this result with the following classes of externality functions.

\subsubsection*{Constant externalities.}
For \emph{constant} externality functions, we may choose $\alpha_j \cdot \Psi_j(x) + \beta_j = G_j(x)$ which makes the externality constraints in \eqref{eq:beckmann} equivalent to~$x$ respecting{} the budget~$B$.
In this case, the condition in \Cref{thm:implementability-for-potential-travel-time}~\ref{thm:implementability-for-potential-travel-time:slater} is trivially fulfilled.
This means, for a given feasible budget, optimal solutions of~\eqref{eq:beckmann} are implementable flows respecting $B$.
The converse, namely that every implementable flow respecting~$B$ optimizes \eqref{eq:beckmann}, holds true if~$B$ is Pareto-minimal:

\begin{restatable}{corollary}{LemCharaPareto}
\label{lem:chara-Pareto}
  Assume~$\tau$ has a convex potential function~$\Phi$, and that all externality functions~$g_{i,p,j}$ are constant, and let~$B$ be a Pareto-minimal budget vector.
  Then, a feasible flow is implementable and respects~$B$ if and only if it solves the optimization problem~\eqref{eq:beckmann} for~$B$.
\end{restatable}
\newcommand{\proofOfLemCharaPareto}{
\begin{proof}
  By \Cref{thm:implementability-for-potential-travel-time}, every optimal solution of~\eqref{eq:beckmann} is implementable and respects~$B$.
  For the converse, let~$x$ be implementable and let $x$ respect~$B$.
  By Pareto-minimality, we have $G(x) = B$.
  Therefore, the linear program~\eqref{eq:lp-characterization-kkt} coincides with \eqref{lp:characterization-implementability}, and thus \Cref{thm:implementability-lp-characterization} and \Cref{thm:implementability-for-potential-travel-time}~\ref{thm:implementability-for-potential-travel-time:kkt-existence-iff-lp} show that $x$ is optimal for~\eqref{eq:beckmann}.
\end{proof}
}
\proofOfLemCharaPareto

Finally, note that if a feasible budget vector~$B$ is not Pareto-minimal, then there might exist implementable flows respecting~$B$ that do not minimize the optimization problem \eqref{eq:beckmann} for~$B$ as illustrated by \Cref{example:implementable-flow-not-solving-opt}.

\begin{example}\label{example:implementable-flow-not-solving-opt}
  Consider the network illustrated in \Cref{fig:implementable-flow-not-solving-opt} with a single externality class and two parallel edges $e_1$ and $e_2$, one edge with a low travel time, say $\tau_{e_1}(x)\coloneqq 0$ and a high externality, $g_{e_1}(x)\coloneqq 1$ and the other with a high travel time $\tau_{e_2}(x)\coloneqq 1$ but low externality $g_{e_2}(x)\coloneqq 0$.
  For any positive budget~$B$, i.e., any budget higher than the minimal feasible budget $B_{\min} = 0$, the flow that only uses the low-externality edge $e_2$ is implementable and respects~$B$; however, it does not solve \eqref{eq:beckmann} for~$B$.
\end{example}

\begin{figure}[ht]
  \centering
  \begin{tikzpicture}
        \tikzset{
        solidnode/.style={
            circle, fill, inner sep=0ex, outer sep=0ex, minimum size=1ex
        }}

        \node[solidnode, label=left:{$s$}] (s) at (0,0) {};
        \node[solidnode, label=right:{$t$}] (t) at (4,0) {};

        \draw[-{stealth}]
            (s) edge[bend left] node[midway, sloped, above] {\footnotesize~$\tau_e\equiv0, g_e\equiv 1$} (t)
            (s) edge[bend right] node[midway, sloped, below] {\footnotesize~$\tau_e\equiv 1, g_e\equiv 0$} (t);

  \end{tikzpicture}
  \caption{A network which admits an implementable flow respecting the minimal feasible budget $0$ that is not an optimal solution of \eqref{eq:beckmann} with respect to any $B > 0$.}
  \label{fig:implementable-flow-not-solving-opt}
\end{figure}

\subsubsection*{Homogeneous externalities.}
Externalities with homogeneous potentials provide a second class for which the budget constraints can be represented exactly.
For $m,n\in\N$, a function $h: \R_{\geq0}^n \to \R^m$ is called \emph{positively homogeneous} of degree $\kappa \geq 0$ if $h(\alpha\cdot x) = \alpha^\kappa \cdot h(x)$ holds for all $\alpha\in\R_{>0}$ and $x\in \R_{\geq0}^n$.
Euler's theorem on homogeneous functions \citep[see][Theorem~20.4]{Simon1994} states that if a function $h$ is positively homogeneous of degree $\kappa$ and differentiable in every coordinate, then $\scalprod{\nabla h_i(x)}{x} = \kappa\cdot h_i(x)$ holds for all $x\in\R_{\geq0}^n$ and $i\in[m]$.

As the following proposition shows, a convex positively homogeneous potential function of an externality function $g_j$ exists if $g_j$ itself is positively homogeneous and has a symmetric and positive semidefinite Jacobian.
The simplest examples of such functions are edge-separable monomial externalities of the form $g_e(x_e) = a_e\cdot x_e^\kappa$.

\begin{proposition}
  Let $g:\R^\mathcal{P}_{\geq0} \to \R^{\mathcal{P}}$ be differentiable and let $\kappa\geq 0$.
  Then $g$ has a convex nonnegative potential, which is positively homogeneous of degree $\kappa + 1$, if and only if $g$ is positively homogeneous of degree $\kappa$ and the Jacobian $\Jac g(x)$ of $g$ is symmetric and positive semidefinite at every $x\in \R^{\mathcal{P}}_{\geq0}$ and, if $\kappa=0$, $g$ is nonnegative.
\end{proposition}
\begin{proof}
  Assume $\Psi$ is a nonnegative convex potential function of $g$ that is positively homogeneous of degree $\kappa+1$.
  Since $g=\nabla\Psi$ is differentiable, the Hessian of $\Psi$, i.e. the Jacobian of $g$, is symmetric and, by convexity, positive semidefinite.
  Let $\alpha > 0$.
  Taking the gradient of $\Psi(\alpha\cdot x) = \alpha^{\kappa + 1} \cdot \Psi(x)$ implies $g(\alpha\cdot x)\cdot \alpha = \alpha^{\kappa + 1}\cdot g(x)$.
  Dividing by $\alpha$ shows that $g$ is positively homogeneous of degree $\kappa$.
  If $\kappa = 0$, then nonnegativity of $g$ follows directly from nonnegativity of $\Psi$.

  Conversely, assume $g$ has a symmetric and positive semidefinite Jacobian and is positively homogeneous of degree $\kappa$.
  Then, the function
  \begin{equation}\label{eq:potential-function-for-homogeneous-externality}
    \Psi(x) \coloneqq \frac{1}{\kappa+1} \cdot \scalprod{g(x)}{x}
  \end{equation}
  is clearly positively homogeneous of degree $\kappa + 1$.
  Furthermore, it is a convex potential of $g$:
  We apply the product rule to obtain
  \begin{align*}
    \nabla \Psi(x) &= \frac{1}{\kappa+1}\cdot \big((\Jac g(x))^\top x + g(x)\big)
    \intertext{and using symmetry of $\Jac g(x)$ and applying Euler's theorem to each component of $g$ yields}
    &= \frac{1}{\kappa +1} \cdot \big( \kappa\cdot g(x) + g(x) \big) = g(x).
  \end{align*}
  The fact that the Jacobian of $g(x)$ is positive semidefinite now implies that $\Psi$ is convex.
  For $\kappa > 0$, nonnegativity of $\Psi$ follows from convexity of $\Psi$ together with $\Psi(0) = 0$ and $g(0) = 0$ (due to homogeneity); for $\kappa=0$ it follows directly from the assumed nonnegativity of~$g$.
\end{proof}

If we assume that the externality functions have suitable convex potentials, then our fixed-point based implementability result, \Cref{thm:smoothness-implementability}, proves implementability of all strictly feasible budgets.
If we additionally assume that $\tau$ has a convex potential, we may employ convex optimization methods to address the problem of computing implementing prices, as we can use \Cref{thm:implementability-for-potential-travel-time}:

\begin{corollary}\label{cor:homogeneous}
  Assume $\tau$ is continuous and each $g_j$ has a convex potential function $\Psi_j$ which is homogeneous (of degree $\kappa_j \geq 0$) for all $j\in J$.
  Then:
  \begin{enumerate}
    \item\label{cor:homogeneous:strictly-feasible-implementable} Every strictly feasible budget vector is implementable.
    \item\label{cor:homogeneous:convex-optimization} If, additionally, $\tau$ has a convex potential, an implementing price vector and a corresponding equilibrium can be computed using convex optimization.
  \end{enumerate}
\end{corollary}
\begin{proof}
  For statement~\ref{cor:homogeneous:strictly-feasible-implementable}, we apply \Cref{thm:smoothness-implementability} by showing that $g_j$ is $(\alpha_j,\beta_j)$-smooth for suitable parameters $\alpha_j,\beta_j$ that have the property $\frac{\alpha_j}{1 - \beta_j}=1$.
  If $g_j$ is the zero-function, we may choose $\alpha_j\coloneqq 1$ and $\beta_j\coloneqq 0$.
  Otherwise, there is some $x\in\R^{\mathcal P}_{\geq0}$ with $\Psi_j(x)\neq 0$, and $\kappa_j > 0$.
  We have $\kappa_j\cdot \Psi_j(x) = \scalprod{g_j(x)}{x} \geq 0$, and thus $\Psi_j(x)$ is strictly positive.
  As $\Psi_j$ is homogeneous and convex, we have $(\nicefrac{1}{2})^{\kappa_j}\cdot\Psi_j(x) = \Psi(x / 2) \leq \nicefrac{1}{2}\cdot \Psi_j(0) + \nicefrac{1}{2}\cdot\Psi_j(x)$.
  As $\Psi_j$ is continuous and homogeneous of positive degree, we have $\Psi_j(0) = 0$, and therefore, dividing by $\Psi_j(x)$ yields $(\nicefrac{1}{2})^{\kappa_j} \leq \nicefrac{1}{2}$, and hence $\kappa_j \geq 1$.

  Thus, we may choose $\alpha_j\coloneqq \frac{1}{\kappa_j} \geq 0$ and $\beta\coloneqq 1-\frac{1}{\kappa_j}\in[0,1)$:
  We use convexity of $\Psi_j$ and the identity $\scalprod{g(x)}{x} = \kappa_j\cdot \Psi_j(x)$ to show
  \begin{align*}
    \scalprod{g(x)}{y} &\leq \scalprod{g(x)}{x} + \Psi_j(y) - \Psi_j(x)
    = \frac{1}{\kappa_j}\cdot\scalprod{g(y)}{y} + (1 - \frac{1}{\kappa_j})\cdot \scalprod{g(x)}{x}.
  \end{align*}

  For statement~\ref{cor:homogeneous:convex-optimization}, we apply  \Cref{thm:implementability-for-potential-travel-time}~\ref{thm:implementability-for-potential-travel-time:slater}:
  As $\Psi_j$ is homogeneous and differentiable, we have $\scalprod{\nabla \Psi_j(x)}{x} = G_j(x)$.
  For $\kappa_j>0$, choose the transformation parameters $\alpha_j,\beta_j$ of \Cref{thm:implementability-for-potential-travel-time} as $\alpha_j\coloneqq\kappa_j$ and $\beta_j\coloneqq0$.
  If $\kappa_j=0$, continuity at the origin implies that $\Psi_j$ is constant, so $g_j=0$; choose $\alpha_j\coloneqq1$ and $\beta_j\coloneqq-\Psi_j(0)$.
  In both cases, $\alpha_j\cdot\Psi_j(x)+\beta_j=G_j(x)$.
\end{proof}

We next consider sums of homogeneous functions; this includes the case of separable externalities $g_e(x_e) = p_e(x_e)$ that are polynomial in the edge load.
For these, we can obtain an implementability factor that is sublinear in the maximum degree, using the fixed-point argument; the factor is the same as the classical price of anarchy bound obtained by \cite{Roughgarden2002b}.
Further, we show that we can use convex optimization to obtain a linear factor in the maximum degree.

\newcommand{\Kappa}{\textnormal{K}}
\begin{corollary}\label{cor:polynomials}
  Assume that $\tau$ is continuous and that each externality function $g_j$ admits a nonempty finite set $\Kappa_j\subseteq\R_{\geq 0}$ and for each $\kappa\in \Kappa_j$ a function $g^{[\kappa]}_j : \R_{\geq0}^{\mathcal{P}} \to \R^{\mathcal{P}}$ with a nonnegative convex potential $\Psi_j^{[\kappa]}$ which is positively homogeneous of degree $\kappa + 1$ such that
  \[
    g_j(x) = \sum_{\kappa\in\Kappa_j} g^{[\kappa]}_j (x)
  \]
  holds for all $x\in \feasFlows$; set $d_j\coloneqq \max\Kappa_j$.
  Let $B$ be a strictly feasible budget vector.
  Then:
  \begin{enumerate}
    \item\label{cor:polynomials:smoothness} The budget $B'\coloneqq (  \frac{1}{1-\beta(d_j)}  \cdot B_j  )_{j\in J}$
      with $\beta(0) \coloneqq 0$ and
      \(
        \beta(d) \coloneqq \frac{d}{(d+1)^{\nicefrac{d + 1}{d}}},
        \text{ for } d > 0,
      \)
      is implementable.
    \item\label{cor:polynomials:convex} If, additionally, $\tau$ has a convex potential, then for the budget $B''\coloneqq \left(  (d_j + 1) \cdot B_j  \right)_{j\in J}$, an implementing price vector and a corresponding Wardrop equilibrium can be obtained using convex optimization.
  \end{enumerate}
\end{corollary}
\begin{proof}
  First note that homogeneity implies $\Psi_j^{[\kappa]}(x) = \frac{\scalprod{g^{[\kappa]}_j(x)}{x}}{\kappa + 1}$ for all $j\in J$ and $\kappa\in\Kappa_j$, and that a convex potential function of $g_j$ is naturally given by
  \begin{equation}\label{eq:potential-function-for-polynomial-externalities}
    \Psi_j(x) \coloneqq \sum_{\kappa\in\Kappa_j} \Psi_j^{[\kappa]}(x).
  \end{equation}
  For statement~\ref{cor:polynomials:smoothness}, we show that $g_j$ is $(1,\beta(d_j))$-smooth.
  For this, first note that as $\Psi_j^{[\kappa]}$ is homogeneous of degree $\kappa + 1$ and convex with gradient $g_j^{[\kappa]}$, it holds for all $s > 0$ that
  \[
    s^{\kappa+1} \cdot \Psi_j^{[\kappa]}(y) = \Psi_{j}^{[\kappa]} (s\cdot y) \geq \Psi_j^{[\kappa]}(x) + \scalprod{g_j^{[\kappa]}(x)}{s\cdot y - x} = s\cdot \scalprod{g_j^{[\kappa]}(x)}{y} - \kappa\cdot \Psi_j^{[\kappa]}(x),
  \]
  where in the last equality we use $\scalprod{g_j^{[\kappa]}(x)}{x} = (\kappa+1)\cdot \Psi_j^{[\kappa]}(x)$.
  For $\kappa > 0$, choosing $s=(\kappa+1)^{1/\kappa}$ yields
  \begin{align*}
    \scalprod{g_j^{[\kappa]}(x)}{y} &\leq
    \bigl((\kappa+1)^{1/\kappa}\bigr)^\kappa \cdot \Psi_j^{[\kappa]}(y) + \frac{\kappa}{(\kappa+1)^{1/\kappa}}\cdot \Psi_j^{[\kappa]}(x)
    \\&=
    (\kappa+1)\cdot \Psi_j^{[\kappa]}(y) + \beta(\kappa)\cdot (\kappa+1) \cdot \Psi_j^{[\kappa]}(x)
    \\
    &= \scalprod{g_j^{[\kappa]}(y)}{y} + \beta(\kappa)\cdot \scalprod{g_j^{[\kappa]}(x)}{x}.
  \end{align*}
  For $\kappa=0$, the same inequality holds with equality.
  Since $\beta(\kappa)$ is nondecreasing in $\kappa$ and all $\Psi_j^{[\kappa]}$ are nonnegative, summing over $\kappa$ proves that $g_j$ is $(1,\beta(d_j))$-smooth.
  The result now follows from \Cref{thm:smoothness-implementability}.

  For statement~\ref{cor:polynomials:convex}, note that as each potential function is nonnegative, we obtain
  \[
    \Psi_j(x) \leq G_j(x) = \sum_{\kappa\in\Kappa_j} \scalprod{g_j^{[\kappa]}(x)}{x}\leq (d_j+1)\cdot \Psi_j(x).
  \]
  In this form, we can again apply \Cref{thm:implementability-for-potential-travel-time}~\ref{thm:implementability-for-potential-travel-time:slater}, instantiated with $\alpha_j\cdot \Psi_j(x) + \beta_j \leftarrow (d_j+1)\cdot \Psi_j(x)$ and $B\leftarrow B''$:
  The resulting problem~\eqref{eq:beckmann} is now feasible as there exists some flow $x$ that strictly respects $(\frac{B''_j}{d_j+1})_j$, and thus $(d_j+1)\cdot \Psi_j(x) \leq (d_j+1)\cdot G_j(x) < B''_j$.
\end{proof}

\subsection{Strict implementability for smooth externalities with explicit prices}

The gap between the two factors in \Cref{cor:polynomials} raises the question of whether prices compatible with the smaller relaxation factor of $1/(1-\beta(d_j))$ can also be obtained through convex optimization.
For multiple externalities, this remains open.
The following theorem answers the question positively for a single externality; in the case of multiple externalities it provides prices under an additional relaxation factor of $|J|$.

Note that this result applies to general smooth externalities with convex potentials under the assumption that a flow strictly respecting the budgets is given.
It provides prices using linear optimization that even \emph{strictly} implement the relaxed budgets, i.e., \emph{every} equilibrium w.r.t. the prices respects the relaxed budgets.

\begin{theorem}\label{thm:explicit-smooth-prices}
  Assume that $\tau$ and $g$ have convex potential functions $\Phi$ and $(\Psi_j)_j$, and that $g_j$ is $(\alpha_j, \beta_j)$-smooth.
  Let $y$ be a given feasible flow that strictly respects a given budget $B$, and 
  let $U\coloneqq \sum_{j\in J} \frac{G_j(y)}{B_j}$.
  Then, for any $\rho > U$, prices that strictly implement the budget $B'\coloneqq (\rho\cdot \frac{\alpha_j}{1-\beta_j}\cdot B_j)_j$ can be obtained using linear optimization; a corresponding equilibrium flow can be obtained using convex optimization.
\end{theorem}
\begin{proof}
  Define $D\coloneqq \max_{x\in\feasFlows} \scalprod{\tau(y)}{y - x}$.
  Define the prices
  \[
    \lambda_j \coloneqq \frac{D + 1}{(\rho - U)\cdot \alpha_j \cdot B_j}
  \]
  for $j\in J$.
  Clearly, the set $\argmin_{x\in\feasFlows} \Phi(x) + \sum_{j\in J}\lambda_j\cdot\Psi_j(x)$ is the set of Wardrop equilibria with respect to prices $\lambda$.
  To prove \emph{strict} implementability, we show that every flow $x$ in this set respects the budget $B'$.

  As Wardrop equilibria fulfill the variational inequality, we have
  \[
    \sum_{j\in J} \lambda_j \scalprod{g_j(x)}{x} \leq \scalprod{\tau(x)}{y - x} + \sum_{j\in J}\lambda_j\scalprod{g_j(x)}{y}.
  \]
  Convexity of $\Phi$ implies monotonicity of $\tau$, and thus $\scalprod{\tau(x)}{y - x}\leq \scalprod{\tau(y)}{y - x} \leq D$.
  For the latter term, we may now use $(\alpha_j,\beta_j)$-smoothness of $g_j$ to get
  \[
    \sum_{j\in J} \lambda_j \scalprod{g_j(x)}{x} \leq D + \sum_{j\in J} \lambda_j\cdot (\alpha_j \cdot \scalprod{g_j(y)}{y} + \beta_j\cdot \scalprod{g_j(x)}{x})
  \]
  Rearranging and dividing by $\frac{D+1}{\rho - U}$ yields
  \[
    \sum_{j\in J} \frac{1-\beta_j}{\alpha_j\cdot B_j}\scalprod{g_j(x)}{x}
    \leq \rho - U + \sum_{j\in J} \lambda_j \cdot \frac{\alpha_j \cdot (\rho - U)}{D + 1} \scalprod{g_j(y)}{y}
    = \rho - U + \sum_{j\in J} \frac{\scalprod{g_j(y)}{y}}{B_j}
    = \rho.
  \]
  As every summand on the left side is nonnegative, we obtain $G_j(x) \leq \rho\cdot \frac{\alpha_j}{1 - \beta_j} \cdot B_j$, concluding the proof.
\end{proof}

\subsection{Monotonicity}

In the following, we examine how total externality responds to changes in the price, with the aim of finding a minimum implementing price for a given budget.
For this, we restrict ourselves to a single externality class with constant externality functions.
Even in this setting, travel times with a convex potential may admit multiple Wardrop equilibria with different total externality values at the same price.

\begin{example}
    Consider a network with two parallel edges~$e_1,e_2$, where edge~$e_1$ has constant travel time~$\tau_{e_1}(x)\coloneqq 1$ and externality~$g_{e_1}=0$ while edge~$e_2$ has travel time $\tau_{e_2}(x)=0$ and externality $g_{e_2}=1$.
    For the price~$\lambda=1$, any split of the flow between the two edges is a Wardrop equilibrium, and the total externality can range from~$0$ to~$d$, where~$d$ denotes the total demand.
\end{example}

Nevertheless, when strictly varying the price, we can show that the total externality behaves in some sense monotonically, as the following lemma shows.
\begin{restatable}{lemma}{ThmMonotonicity}
\label{thm:monotonicity}
  Assume~$|J| = 1$ and that~$\tau$ has a convex potential~$\Phi$ and all~$g_{i,p}$ are constant.
  Then, for every~$\lambda_1 < \lambda_2$ with Wardrop equilibria~$x_1\in\WE(\lambda_1)$ and~$x_2\in\WE(\lambda_2)$, it holds~$G(x_1) \geq G(x_2)${} and~$\Phi(x_1) \leq \Phi(x_2)$.
\end{restatable}
\newcommand{\proofOfThmMonotonicity}{
\begin{proof}
  By \Cref{prop:wardrop-iff-minimizes-potential}, we know that~$x_1$ and~$x_2$ minimize~$\Phi^{\lambda_1}$ and~$\Phi^{\lambda_2}$, respectively.
  This implies \begin{align*}
    \Phi(x_1) + \lambda_1 \cdot G(x_1) &\leq \Phi(x_2) + \lambda_1 \cdot G(x_2), \quad \text{and} \\
    \Phi(x_2) + \lambda_2 \cdot G(x_2) &\leq \Phi(x_1) + \lambda_2\cdot G(x_1).
  \end{align*}
  Equivalently, we have \begin{align*}
    \Phi(x_1) - \Phi(x_2) &\leq \lambda_1 \cdot (G(x_2) - G(x_1)), \quad \text{and} \\
    \Phi(x_1) - \Phi(x_2) &\geq \lambda_2 \cdot (G(x_2) - G(x_1)).
  \end{align*}
  This implies~$(\lambda_2 - \lambda_1) \cdot (G(x_2) - G(x_1)) \leq 0$, and thus~$G(x_1) \geq G(x_2)$.
  This also means that~$\Phi(x_1) - \Phi(x_2) \leq \lambda_1 \cdot (G(x_2) - G(x_1)) \leq 0$, and hence~$\Phi(x_1) \leq \Phi(x_2)$.
\end{proof}
}
\proofOfThmMonotonicity

\begin{remark}
  Assume that we have multiple externality classes with a weight~$w_j\in\R_{\geq0}$ per class,  and that we are interested in the weighted externality of the externality classes, i.e.,~$\smash{\tilde G(x) \coloneqq \scalprod{w}{G(x)}}$.
  Then, we can apply \Cref{thm:monotonicity} to the weighted externality factors~$\smash{\tilde g_{i,p} \coloneqq \sum_{j\in J} w_j \cdot g_{i,p,j}}$, and obtain a similar monotonicity result when scaling the price vector given by~$w$:
  For any~$\alpha_1<\alpha_2$ and any Wardrop equilibria~$x_1$ and~$x_2$ with respect to~$c^{\alpha_1\cdot w}$ and~$c^{\alpha_2\cdot w}$, it holds that~$\tilde G(x_1) \geq \tilde G(x_2)$ and~$\Phi(x_1)\leq \Phi(x_2)$.
\end{remark}
This monotonicity of the total externality allows a simple binary search for the minimal~$\lambda$ for which a corresponding Wardrop equilibrium respects a given budget~$B$.
For this, the proof of \Cref{cor:implementability-single-resources} provides an explicit upper bound for the minimal~$\lambda$ which can be used to initialize the search.
Note that in every update of the binary search, we have to solve the convex minimization problem~$\min \Phi^\lambda(x)$ over the set of feasible flows.
As we show in \Cref{subsec:budget-binary}, this binary search can be implemented in polynomial time if the travel time functions are separable and piecewise affine.

The constant-externality assumption is essential for this monotonicity guarantee.
For non-constant externality functions, even affine ones, total externality can increase with the price as the following example inspired by Braess' paradox~\citep{Braess68} demonstrates.
\newcommand{\exampleNonMonotonicity}{
\begin{example}\label{example:non-monotonicity}
  We consider the Braess-inspired network depicted in \Cref{fig:braess-non-constant-consumption}:
  Setting a price of~$\lambda > \nicefrac{1}{2}$, the unique equilibrium flow~$x$ exclusively uses the zig-zag-path~$p= (s,v,w,t)$ with a cost of~$c_p^\lambda(x) = 4\lambda+\frac{1}{4} < \frac{9}{2}\lambda$ and a total externality of~$G(x) = 8$.
  For any~$\lambda$ with~$0 < \lambda < \nicefrac{1}{6}$, however, the unique equilibrium flow~$y$ uses the upper path~$(s,v,t)$ and the lower path~$(s,w,t)$ with 1 unit of flow each, as the difference in cost of any of these two paths and the zig-zag-path is~$\lambda\cdot\nicefrac{7}{2} - (\nicefrac{1}{4}+\lambda\cdot 2) = \lambda\cdot\nicefrac{3}{2}-\nicefrac{1}{4} < 0$.
  Nevertheless, the total externality of~$y$ is~$G(y) = 7$.
\end{example}
\begin{figure}[ht]
  \centering
  \begin{minipage}[c]{0.51\linewidth}
  \centering
  \resizebox{\linewidth}{!}{%
  \begin{tikzpicture}
        \tikzset{
        solidnode/.style={
            circle, fill, inner sep=0ex, outer sep=0ex, minimum size=1ex
        }}

        \node[solidnode, label=left:{$s$}] (s) at (0,0) {};
        \node[solidnode, label=above:{$v$}] (v) at (2,1) {};
        \node[solidnode, label=below:{$w$}] (w) at (2,-1) {};
        \node[solidnode, label=right:{$t$}] (t) at (4,0) {};

        \draw[-{stealth}]
            (s) edge node[midway, sloped, above] {\footnotesize~$0+\lambda\cdot x_e$} (v)
            (s) edge node[midway, sloped, below] {\footnotesize~$0+\lambda\cdot\frac{5}{2}$}  (w)
            (v) edge node[midway, right] {\footnotesize~$\frac{1}{4} + \lambda \cdot 0$} (w)
            (v) edge node[midway, sloped, above] {\footnotesize~$0+\lambda\cdot\frac{5}{2}$}  (t)
            (w) edge node[midway, sloped, below] {\footnotesize~$0+\lambda\cdot x_e$}  (t);

        \draw[decorate, decoration={snake, amplitude=0.5mm, segment length=4mm}, -{stealth}]
            (-2.5,0) -- (-.75, 0) node[midway, above] {\footnotesize~$d=2$};
  \end{tikzpicture}
  }
  \end{minipage}\hfill
  \begin{minipage}[c]{0.46\linewidth}
  \centering
  \begin{tikzpicture}
    \begin{axis}[
      width=\linewidth, height=0.5\linewidth,
      grid=major, tick label style={font=\footnotesize},
      label style={font=\footnotesize},
      xlabel={Price $\lambda$}, ylabel={Total externality $G(x)$},
      xmin=0, xmax=0.6666666666666667, ymin=6.8, ymax=8.1,
      xtick={0, 0.1666666666666667,0.3333333333333333,0.5,0.6666666666666667},
      xticklabels={0,$\frac{1}{6}$,$\frac{1}{3}$,$\frac{1}{2}$,$\frac{2}{3}$},
      extra x ticks={0.2},
      extra x tick labels={$\frac{1}{5}$},
      extra x tick style={grid=none},
      ytick={7,7.5,8},
      scaled ticks=false
    ]
      \addplot[-] table [x=price, y=total_consumption, col sep=comma]
        {data/non-monotone-braess-carbon-pricing.csv};
    \end{axis}
  \end{tikzpicture}
  \end{minipage}
  \caption{A Braess-inspired network (left) whose edges are labelled with~$\tau_e(x_e) + \lambda\cdot g_e(x_e)$.
  The total externality for the equilibria $x$ (right) is not monotonically decreasing in the price~$\lambda$.
  }
  \label{fig:braess-non-constant-consumption}
\end{figure}
}
\exampleNonMonotonicity

\newcommand{\numbp}{K}
\newcommand{\BP}{\sigma}
\newcommand{\bp}{k}
\newcommand{\aparts}{\ell}  %
\newcommand{\estate}{\ensuremath{\mathcal{S}}}

\section{Separable piecewise affine travel time functions}
\label{sec:affine}
In this section, we consider the special case of one externality class, i.e., $|J| = 1$, and cost functions~$c^{\lambda}$ consisting of separable, piecewise affine travel time functions and separable, constant externality costs.
This means the cost of a path~$p$ can be written as~$\smash{c^{\lambda}_{i,p}(x) = \sum_{e \in p} c^{\lambda}_{e} (x_e)}$, where~$c_e^{\lambda}$ is a commodity-independent edge cost function depending on the total edge load~$\smash{x_e\coloneqq \sum_{(i,p)\in\mathcal P} \mathbf{1}_{e\in p} \cdot x_{i,p}}$ on edge~$e\in E$ and the externality price~$\lambda \geq 0$.
Additionally, we assume that the edge cost functions are continuous, piecewise affine, and strictly increasing. 
Formally, we assume that every edge cost function has breakpoints~$0 = \BP_{e, 0} < \BP_{e, 1} < \dots < \BP_{e, \numbp} < \BP_{e, \numbp + 1} = \infty$, such that the cost of edge~$e \in E$ is
\begin{equation}\label{eq:parametric_comp:costfunctions}
c^{\lambda}_{e} (x_e) = a_{e, \bp} x_e + b_{e, \bp} + g_e \lambda
\quad \text{for all } x_e \in [\BP_{e, \bp}, \BP_{e, \bp+1}) \text{{} and } \bp \in \zeroto{\numbp}
,
\end{equation}
where~$a_{e, \bp} > 0$,~$b_{e, \bp} \in \Q$, and~$g_e \geq 0$ are rational constants such that the resulting function is nonnegative (for $\lambda=0$), continuous and strictly increasing.
Note that we assume that every cost function has~$\numbp + 1$ parts.
This is without loss since any function part may be subdivided without changing its coefficients.

In this setting, we will show how to compute a complete curve of Wardrop equilibria depending on the price~$\lambda$ of the single externality class.
Since we are dealing with separable cost functions, we express flows as matrices~$x = (x_{i,e}) \in \R^{I \times E}_{\geq 0}$ consisting of the amounts of flow~$x_{i,e}$ that commodity~$i \in I$ sends on edge~$e \in E$. 
By assumption, the cost functions are strictly increasing and, therefore, the edge loads~$x_e \coloneqq \sum_{i \in \setCommodities} x_{i,e}$ of a Wardrop equilibrium are unique for every fixed~$\lambda \geq 0$.
Our formal objective of this section is to compute a function %
\optDisplay{
    x : \R_{\geq 0} \to \R^{\setCommodities\times E}, \lambda \mapsto x(\lambda) = (x_{i,e}(\lambda)),
}
where~$x(\lambda) \in \WE(\lambda)$ is a Wardrop equilibrium with respect to the cost functions given in~\eqref{eq:parametric_comp:costfunctions}.  
We state this in the following theorem.

\begin{theorem}
  \label{thm:piecewise:computation}
    For a single externality class and separable, piecewise affine, continuous and strictly increasing cost functions~$c^{\lambda}_e$ of the form in~\eqref{eq:parametric_comp:costfunctions},
    there is a continuous, piecewise affine function~$\lambda \mapsto x(\lambda)$ mapping every~$\lambda \geq 0$ to a Wardrop equilibrium flow~$x(\lambda) \in \WE(\lambda)$ with respect to the cost functions~$c^{\lambda}_e$. Furthermore, the function has~$\mathcal{O}(2^{|\setCommodities| |E|} (\numbp+1)^{|E|})$ breakpoints.
    For rational input data, the function~$\lambda \mapsto x(\lambda)$ can be computed exactly in finite time.
\end{theorem}

Before proving \Cref{thm:piecewise:computation}, we state the following, well-known characterization of Wardrop equilibria via optimal node potentials.
The result follows immediately from the KKT conditions of the corresponding optimization problem of the potential function as discussed in \Cref{prop:wardrop-iff-minimizes-potential}.
For a proof, see, for example,~\citet{KlimmWarode2021}.

\begin{lemma}\label{lem:potentials}
    For every fixed~$\lambda \geq 0$, a feasible flow~$x \in \feasFlows$ is a Wardrop equilibrium if and only if for every commodity~$i \in \setCommodities$ there exists an optimal potential vector~$\smash{\pi^{(i)} = (\pi_{i, v}) \in \R^V}$ such that
    \begin{align*}
        c^{\lambda}_{e} (x_e) &= \pi_{i, w} - \pi_{i, v} \qquad \text{if } x_{i, e} > 0 \\
        c^{\lambda}_{e} (x_e) &\geq \pi_{i, w} - \pi_{i, v} \qquad \text{if } x_{i, e} = 0
    \end{align*}
    for every edge~$e = vw \in E$, where~$x_e$ is the edge load of the edge.
\end{lemma}

The optimal potentials from \Cref{lem:potentials} are not unique in general.
Therefore, we denote for~$i \in \setCommodities$ and~$v \in V$ by~$\phi_{i, v} (\lambda)$ the length of a shortest path from~$s_i$ to~$v$ with respect to the edge lengths~$c^{\lambda} (x_e)$, where~$x_e$ is the edge load of a Wardrop equilibrium, if such a path exists; otherwise the exact value does not matter and can be set to the maximum of all other values.
Since the corresponding potential function is strictly convex with respect to edge loads, these edge loads, and therefore also the shortest path potentials~$\phi_{i, v} (\lambda)$, are then unique and well-defined.
It is easy to verify that for every~$x \in \WE(\lambda)$, the shortest path lengths~$\phi_{i, v} (\lambda)$ form an optimal potential vector in the sense of \Cref{lem:potentials}.
Given~$\lambda \geq 0$, we say an edge~$e = vw \in E$ is \emph{active for commodity~$i \in \setCommodities$} if $e$ is reachable from $s_i$ and ~$c_e^\lambda(x_e) = \phi_{i, w} (\lambda) - \phi_{i, v} (\lambda)$, i.e., if the edge lies on a shortest path for the given commodity.
By the definition of the piecewise affine edge cost functions for every edge, there is an index~$\bp$ such that~$\BP_{e, \bp} \leq x_e < \BP_{e, \bp + 1}$. We say that the function part~$\bp$ is \emph{active} for edge~$e$.

Before we prove \Cref{thm:piecewise:computation}, we first consider the setting of a fixed externality price~$\lambda \geq 0$. In this case, we prove in the following lemma that a Wardrop equilibrium can be computed efficiently. The main idea is to use \cref{prop:wardrop-iff-minimizes-potential}, i.e., to obtain a Wardrop equilibrium as a minimizer of the Beckmann potential~$\Phi$. However, for our setting of piecewise affine cost functions, the Beckmann potential~$\Phi$ has a piecewise quadratic objective function. In order to show that a Wardrop equilibrium can be computed in polynomial time anyway, we give a slightly different convex quadratic program with an affine linear objective function that still yields a minimizer of the Beckmann potential~$\Phi$.

\begin{restatable}{lemma}{LemPiecewiseFixedComputation}
\label{lem:piecewise:fixed:computation}
  Given separable, piecewise affine, continuous, and strictly increasing edge cost functions~$c_e^{\lambda}$ with rational coefficients and a single externality class, and some rational~$\lambda \geq 0$, a Wardrop equilibrium~$x \in \WE (\lambda)$ can be computed in polynomial time.
\end{restatable}
\newcommand{\proofOfLemPiecewiseFixedComputation}{
\begin{proof}
    Consider the Beckmann--McGuire--Winsten (BMW) potential
    \[
    \Phi(x) = \sum_{e \in E} \int_0^{x_e} c^{\lambda}_e (t) \diff t.
    \]
    By \cref{prop:wardrop-iff-minimizes-potential}, every feasible flow minimizing this potential is also a Wardrop equilibrium. We will prove that a feasible flow minimizing the BMW potential can be computed using a quadratic program with smooth objective and affine constraints.
    
    For every edge $e \in E$ and every piece $k \in \zeroto{\numbp}$, let $u_{e, k} \coloneqq \sigma_{e, k+1} - \sigma_{e, k}$ and $u_{e, K} \coloneqq \infty$. The idea is to replace the total loads $x_e$ by $K+1$~variables $y_{e, k} \in [0, u_{e, k}]$ that represent the amount of flow in the $k$th segment of the piecewise affine function $c^{\lambda}_e$. Define
    \[
        q^{\lambda}_{e, k} (y) \coloneqq \int_{0}^y c^{\lambda}_e (\sigma_{e,k} + t) \diff t = \frac{a_{e, k}}{2} y^2 + c^{\lambda}_{e} (\sigma_{e,k}) y
        \qquad
        \text{for } 0 \leq y \leq u_{e, k} \text{ and } k \in \zeroto{K}
        .
    \]
    Then, consider the following quadratic program:
    \begin{equation}
    \begin{aligned}
      \min_{x,y}~& \sum_{e \in E} \sum_{k = 0}^K q^{\lambda}_{e, k} (y_{e, k})\\
      \text{s.t.}~& 
      \sum_{i \in I} x_{i,e} = \sum_{k=0}^K y_{e, k} &&\text{for } e \in E\\
      & 0 \leq y_{e,k} \leq u_{e,k} &&\text{for } e \in E, k\in \zeroto{K-1}, \\
      & y_{e, K} \geq 0 &&\text{for } e \in E, \\
      & x \in \feasFlows.
    \end{aligned}\tag{QP$_\lambda$}\label{eq:qp_lambda}
    \end{equation}
    The quadratic program~\eqref{eq:qp_lambda} has a convex quadratic objective function and affine constraints. Further, all coefficients are rational such that, by appropriate scaling, we can assume that all coefficients are integer. Therefore, \eqref{eq:qp_lambda} can be solved in polynomial time~\citep{kozlov1980}. It remains to be shown that a solution of~\eqref{eq:qp_lambda} also minimizes the Beckmann potential~$\Phi$.
    First, we claim that for every edge~$e$ and every possible (total edge load) $x \geq 0$, we have
    \begin{equation}
        \int_{0}^x c^{\lambda}_e (t) \diff t = \min \bigg\{ \sum_{k = 0}^K q^{\lambda}_{e, k} (y_{e, k}) \;\bigg\vert\; \sum_{k=0}^K y_{e, k} = x, 0 \leq y_{e,k} \leq u_{e,k} \text{ for } k \in \zeroto{K-1}, y_K \geq 0 \bigg\}
        \label{eq:piecewise-edge-load},
    \end{equation}
    i.e., every splitting of a total edge load~$x$ into loads of the pieces that minimizes~$\sum_{k = 0}^K q^{\lambda}_{e, k} (y_{e, k})$ exactly gives the value that edge~$e$ contributes to the BMW potential~$\Phi$.
    In order to prove~\eqref{eq:piecewise-edge-load}, we first claim that if~$y$ is a minimizer in~\eqref{eq:piecewise-edge-load}, we have $y_{e, k} = u_{e, k}$ for $k < r$, $y_{e, k} = x - \sigma_{e, r}$ if $k = r$, and $y_{e, k} = 0$ if $k > r$, where~$r$ is the unique index with $\sigma_{e,r} \leq x < \sigma_{e, r+1}$. 
    Indeed, assume that was not the case. Then there are indices $k < \ell$ with $y_{e, k} < u_{e, k}$ and $y_{\ell} > 0$. By continuity and strict increase of $c^{\lambda}_e$, we get
    \[
        (q^{\lambda}_{e, k})'(y_{e,k}) = c^{\lambda}_e (\sigma_{e,k}) + a_{e,k} y_{e,k} = c^{\lambda}_e (\sigma_{e, k} + y_{e,k}) < c^{\lambda}_{e} (\sigma_{e, k+1}) \leq c^{\lambda}_{e} (\sigma_{e, \ell}) < (q^{\lambda}_{e, \ell})' (y_{\ell}).
    \]
    Thus, moving a sufficiently small amount of load from segment~$\ell$ to segment~$k$ strictly decreases the objective and $y$ cannot be a minimizer. Therefore, for every minimizer of \eqref{eq:piecewise-edge-load}, every earlier segment is full before a later segment carries positive load.
    Finally, $\sum_{k < r} y_{e, k} = \sigma_{e,r}$ and continuity at breakpoints, prove~\eqref{eq:piecewise-edge-load}:
    \[
    \sum_{k=0}^K q^{\lambda}_{e, k} (y_{e, k}) = 
    \sum_{k < r} \int_{0}^{u_{e,k}} c^{\lambda}_e (\sigma_{e,k} + t) \diff t +
    \int_{0}^{x - \sigma_{e,r}} c^{\lambda}_e (\sigma_{e,r} + t) \diff t =
    \int_{0}^x c^{\lambda}_e (t) \diff t.
    \]

    Overall, for every feasible flow $x \in \feasFlows$, we can find by~$\eqref{eq:piecewise-edge-load}$ piecewise loads $y$ such that $(x, y)$ are feasible in~\eqref{eq:qp_lambda} with objective $\sum_{e \in E} \sum_{k = 0}^K q^{\lambda}_{e, k} (y_{e, k}) = \Phi(x)$. 
    Conversely, for every feasible pair $(x,y)$ in~\eqref{eq:qp_lambda}, we have $\sum_{e \in E} \sum_{k = 0}^K q^{\lambda}_{e, k} (y_{e, k}) \stackrel{\eqref{eq:piecewise-edge-load}}{\geq} \Phi(x)$.
    It follows that the optimal value of \eqref{eq:qp_lambda} equals the minimal BMW potential over all feasible flows. Moreover, the $x$-component of every optimal solution of~\eqref{eq:qp_lambda} minimizes the BMW potential: otherwise, a flow with lower potential would give a better solution of~\eqref{eq:qp_lambda}. 
    By \cref{prop:wardrop-iff-minimizes-potential}, minimizing~\eqref{eq:qp_lambda} in polynomial time yields a Wardrop equilibrium $x$.
\end{proof}
}
\proofOfLemPiecewiseFixedComputation
As edge loads are unique for all~$\lambda \geq 0$, the set of active edges is fixed for every commodity. 
We define by 
\begin{align}
\label{eq:support}
    S(\lambda) \coloneqq \{ (i, e) \in \setCommodities \times E \mid \text{$e$ is reachable from $s_i$ and } c^\lambda(x_e) = \phi_{i, w} (\lambda) - \phi_{i, v} (\lambda) \}
\end{align}
the \emph{active support for~$\lambda \geq 0$}.
Note that this definition differs from the standard definition of supports, where an edge is usually in the support for some commodity if the flow on the edge is non-zero.
For the same reasons, for every edge~$e$ there is a unique \emph{active function part for $\lambda$}, i.e., there is $\aparts_e (\lambda) \in \zeroto{\numbp}$ such that $x_e (\lambda) \in [\BP_{e, \aparts}, \BP_{e, \aparts+1})$.
We denote by
\begin{align}
\label{eq:active-parts}
    \aparts (\lambda) \coloneqq (\aparts_{e_1}, \dotsc, \aparts_{e_m}) \in \zeroto{\numbp}^E
    \quad \text{such that} \quad
    x_e \in [\BP_{e, \aparts}, \BP_{e, \aparts+1})
    \quad \text{for all } e \in E
\end{align}
the \emph{vector of active function parts} for~$\lambda \geq 0$.
We refer to the tuple~$\estate (\lambda) \coloneqq (S(\lambda), \aparts(\lambda))$ consisting of the support and the active function parts as the \emph{edge state} for~$\lambda \geq 0$. 
Given an arbitrary edge state~$\estate = (S, \aparts)$ with $S \subseteq \setCommodities \times E$ and $\aparts \in \zeroto{\numbp}^E$,
we define the set
\begin{align*}
P(\estate) \coloneqq
\bigl\{ (x, \pi, \lambda) \in \feasFlows \times \R^{\setCommodities \times V} \times \R_{\geq 0} \;\big\vert\;
\pi_{i, w} - \pi_{i, v} &= c_e^{\lambda} (x_e) &&\text{for all } e = vw \text{ with } (i, e) \in S \notag\\
\pi_{i, w} - \pi_{i, v} &\leq c_e^{\lambda} (x_e) &&\text{for all } e = vw \text{ with } (i, e) \notin S 
\\
x_{i,e} &= 0 &&\text{for all } e = vw \text{ with } (i, e) \notin S  \notag\\
\BP_{e, \aparts_e} \leq x_{e} &\leq \BP_{e, \aparts_e + 1} &&\text{for all } e \in E \notag %
        \bigr\}
\end{align*}
containing all triples of flows, potentials and a price that satisfy \cref{lem:potentials} and are consistent with the edge state $\estate$ as we prove in the following lemma.
\begin{lemma}\label{lem:linear-wardrop-polyhedron}
    For every edge state~$\estate = (S, \ell)$, the following holds.
    \begin{enumerate}
        \item The set $P(\estate)$ is a convex polyhedron.
        \item Every triple $(x, \pi, \lambda) \in P(\estate)$ consists of a Wardrop equilibrium $x \in \WE(\lambda)$ for the price~$\lambda$, an optimal potential~$\pi$ and the price~$\lambda$.
        \item Let $\lambda \geq 0$, $x \in \WE(\lambda)$ be some Wardrop equilibrium and $\phi(\lambda)$ be the corresponding shortest path potential. Then, $(x, \phi(\lambda), \lambda) \in P(\estate)$ for $\estate = (S(\lambda), \ell(\lambda))$.
    \end{enumerate}
\end{lemma}
\begin{proof}
    For fixed~$\estate$, the last set of inequalities forces the total edge load to be between two consecutive breakpoints of the piecewise affine cost function. Hence, the cost function can be explicitly written as $c^{\lambda}_e(x_e) = a_{e, \aparts_e} x_e + b_{e, \aparts_e} + g_e \lambda$. Therefore, $P(\estate)$ is given by a set of linear constraints and, hence, is a convex polyhedron, proving statement 1.

    Let $(x, \pi, \lambda) \in P(\estate)$. Then $x$ is a feasible flow and together with~$\pi$ and $\lambda$ satisfy the conditions of \cref{lem:potentials} showing statement 2. Statement 3 follows again from \cref{lem:potentials} and the fact that the shortest path lengths are an optimal potential. The final inequalities hold by the definition of~$\ell(\lambda)$.
\end{proof}
We denote by 
\[ \Lambda (\estate) \coloneqq \{ \lambda \geq 0 \mid \exists\, (x, \pi) :  (x, \pi, \lambda) \in P(\estate) \}\]
the (possibly empty) set of externality prices~$\lambda \geq 0$, for which the active support of some equilibrium is consistent with~$S$ and the active function parts are consistent with~$\aparts$.
Note that $\Lambda(\estate)$ is a projection of the polyhedron~$P(\estate)$ and therefore convex and closed.
Hence, the set~$\Lambda(\estate)$ is a (possibly empty) interval that can be computed efficiently as shown in the following lemma.

\begin{lemma}
  \label{lem:active-supports-interval}
    Let~$\estate = (S, \aparts)$ be an arbitrary edge state. 
    If~$\Lambda(\estate)$ is nonempty, then $\Lambda(\estate)$ is a closed interval and we can compute the values~$\underline{\lambda} \coloneqq \min \Lambda(\estate)$ and~$\smash{\overline{\lambda}} \coloneqq \max \Lambda(\estate)$ (or assert that $\sup \Lambda(\estate) = \infty$) together with Wardrop equilibria~$\underline{x} \in \WE(\underline{\lambda})$ and~$\overline{x} \in \WE(\overline{\lambda})$ (if $\sup \Lambda(\estate) < \infty$) in polynomial time.
\end{lemma}
\begin{proof}
    Let~$\estate = (S, \aparts)$ with~$S \subseteq \setCommodities \times E$ and~$\aparts \in \zeroto{\numbp}^E$ be fixed.
    Consider the following linear program.
    \begin{equation}\tag{LP$_{\mathcal{S}}$}\label{eq:active-support:lp}
      \begin{aligned}
          \max ~& \lambda\\
          \text{s.t.} ~& (x, \pi, \lambda) \in P(\estate)
      \end{aligned}
    \end{equation}

    As argued in the proof of \cref{lem:linear-wardrop-polyhedron}, all constraints in $P(\estate)$ are linear. Hence, \eqref{eq:active-support:lp} is indeed a linear program. Further, \eqref{eq:active-support:lp} is infeasible only if~$\Lambda(\estate)$ is empty. 
    If \eqref{eq:active-support:lp} has an optimal solution $(\overline{x}, \overline{\pi}, \overline{\lambda})$, then $\overline{\lambda} = \max \Lambda(\estate)$ by definition. Further, by \cref{lem:linear-wardrop-polyhedron}, $\overline{x} \in \WE(\overline{\lambda})$.
    If \eqref{eq:active-support:lp} is unbounded, then~$\sup \Lambda (\estate) = \infty$ by definition.
    Additionally,~$\underline{\lambda}$ and~$\underline{x}$ can be computed with the minimization variant of \eqref{eq:active-support:lp}. Since~$\lambda \geq 0$, the LP cannot be unbounded in this case.
 \end{proof}

Consider a fixed edge state $\estate = (S, \aparts)$ with $\Lambda(\estate) \neq \emptyset$ and $\overline{\lambda} < \infty$. Let $\underline{\lambda}, \overline{\lambda}, \underline{x}, \overline{x}$ be defined as in \cref{lem:active-supports-interval}. Additionally, let $\underline{\pi}$ and $\overline{\pi}$ be optimal potentials for the respective equilibria. Then $(\underline{x}, \underline{\pi}, \underline{\lambda}), (\overline{x}, \overline{\pi}, \overline{\lambda}) \in P(\estate)$. Since $P(\estate)$ is convex, every convex combination of $(\underline{x}, \underline{\pi}, \underline{\lambda})$ and $(\overline{x}, \overline{\pi}, \overline{\lambda})$ is also in $P(\estate)$ and, by \cref{lem:linear-wardrop-polyhedron} the corresponding flow is a Wardrop equilibrium for all $\lambda \in [\underline{\lambda}, \overline{\lambda}]$. So we can define the following function mapping $\lambda$ to the corresponding convex combination $x(\lambda) \in \WE(\lambda)$ as
\begin{equation}\label{eq:equilibrium-curve}
    x(\lambda) \coloneqq
    \frac{\lambda - \underline{\lambda}}{\overline{\lambda} - \underline{\lambda}} \, \overline{x} + 
    \frac{\overline{\lambda} - \lambda}{\overline{\lambda} - \underline{\lambda}} \, \underline{x} 
    \qquad \text{for } \lambda \in [\underline{\lambda}, \overline{\lambda}]
\end{equation}
if $\underline{\lambda} < \overline{\lambda}$. If $\Lambda(\estate)$ is a singleton, we set $x(\lambda) = \overline{x}$ for $\lambda \in \Lambda(\estate)$.

In order to complete the proof of \Cref{thm:piecewise:computation}, we need to find all relevant edge states~$\estate$ such that the corresponding functions defined in~\eqref{eq:equilibrium-curve} cover all of the interval $[0, \infty)$. In particular, we need a last function part for large values of~$\lambda$ that can be extended to infinity.

\begin{proof}[Proof of \Cref{thm:piecewise:computation}]
    Let $\lambda_{\max} \coloneqq Q\cdot \sum_{e\in E}\tau_e(D)+1$, where~$D\coloneqq \sum_{i\in I} d_i$ is the total demand and where~$Q$ is the product of the denominators of all~$g_e$.
    Then, \Cref{cor:implementability-single-resources} guarantees that for~$\lambda\geq\lambda_{\max}$ only paths with minimal externality are used in an equilibrium.
    On minimum-externality paths, increasing the price adds the same cost to every such path of a commodity.
    Hence, any equilibrium at $\lambda_{\max}$ remains an equilibrium for all larger prices, and we can choose $x(\lambda)$ to be constant on $[\lambda_{\max},\infty)$.

    The high-level idea of the computation is as follows. We iteratively obtain valid edge states by computing equilibria for given fixed prices~$\lambda_{\text{l}}$ and $\lambda_{\text{r}}$. For these edge states, we obtain intervals $[\underline{\lambda}, \overline{\lambda}]$ that can be computed with the linear program~\eqref{eq:active-support:lp} from the proof of~\cref{lem:active-supports-interval}.
    A Wardrop equilibrium on these intervals is then given by
    the function defined in~\eqref{eq:equilibrium-curve}.
    Therefore, we have computed two of the function parts of the desired function $\lambda \mapsto x(\lambda)$.
    Then, we iterate this process, using the midpoint between $\lambda_{\text{l}}$ and $\lambda_{\text{r}}$.
    
    We start with $\lambda_{\text{l}} = 0$ and $\lambda_{\text{r}} = \lambda_{\max}$.
    For each value, we compute a Wardrop equilibrium, which is possible in polynomial time by \Cref{lem:piecewise:fixed:computation}. 
    The equilibrium also yields the respective edge states~$\estate_{\text{l}}$ and~$\estate_{\text{r}}$. 
    Then, we use the linear program~\eqref{eq:active-support:lp} from \cref{lem:active-supports-interval} to compute the intervals $[\underline{\lambda}_{\text{l}}, \overline{\lambda}_{\text{l}}] \ni \lambda_{\text{l}}$ and $[\underline{\lambda}_{\text{r}}, \overline{\lambda}_{\text{r}}] \ni \lambda_{\text{r}}$ together with the respective Wardrop equilibria $\underline{x}_{\text{l}}, \overline{x}_{\text{l}}, \underline{x}_{\text{r}}, \overline{x}_{\text{r}}$ in polynomial time. 
    If the intervals $[\underline{\lambda}_{\text{l}}, \overline{\lambda}_{\text{l}}]$ and $[\underline{\lambda}_{\text{r}}, \overline{\lambda}_{\text{r}}]$ intersect or are the same, we have covered the interval between $\underline{\lambda}_{\text{l}}$ and~$\overline{\lambda}_{\text{r}}$.
    Otherwise, let $\lambda_{\text{m}} \coloneqq (\overline{\lambda}_{\text{l}} + \underline{\lambda}_{\text{r}})/2$, and denote by $\estate_{\text{m}}$ the edge state for $\lambda = \lambda_{\text{m}}$.
    Again, with \eqref{eq:active-support:lp}, we compute the interval $[\underline{\lambda}_{\text{m}}, \overline{\lambda}_{\text{m}}] \ni \lambda_{\text{m}}$.
    We iterate this process recursively for $\lambda'_{\text{l}} = \overline{\lambda}_{\text{l}}$ and $\lambda'_{\text{r}} = \underline{\lambda}_{\text{m}}$ as well as $\lambda''_{\text{l}} = \overline{\lambda}_{\text{m}}$ and $\lambda''_{\text{r}} = \underline{\lambda}_{\text{r}}$ as long as $\lambda'_{\text{l}} < \lambda'_{\text{r}}$ and $\lambda''_{\text{l}} < \lambda''_{\text{r}}$, respectively.
    We retain the computed intervals and recurse only on still-uncovered gaps.
    Each sampled midpoint therefore lies outside all previously computed intervals and yields a new edge state.
  
    After at most as many midpoint steps as there are feasible edge states, the computed intervals cover $[0,\lambda_{\max}]$.
    For some breakpoints, there might be multiple Wardrop equilibria, since the intervals $[\underline{\lambda}, \overline{\lambda}]$ intersect in the breakpoints.
    By choosing one Wardrop equilibrium per breakpoint arbitrarily, we obtain pairs of breakpoints $\lambda$ and Wardrop equilibria $x(\lambda)$. We can then interpolate the Wardrop equilibria between these breakpoints with the convex combination as in \eqref{eq:equilibrium-curve}. The resulting function is then continuous and piecewise affine by definition.
    
    To compute the very last part of the equilibrium curve (on the interval $(\lambda_{\max}, \infty)$), set $x(\lambda)=x(\lambda_{\max})$, which is a Wardrop equilibrium by the argument above.

    Since there are at most as many function parts as there are possible edge states, it follows that the equilibrium curve has at most~$\mathcal{O}(2^{|\setCommodities||E|}(\numbp+1)^{|E|})$ many function parts.
\end{proof}

\subsection{Networks with equilibrium curves with exponentially many edge states}

\Cref{thm:piecewise:computation} establishes that the function~$\lambda\mapsto x(\lambda)$ describing the Wardrop equilibrium as a function of the parameter~$\lambda$ has at most an exponential number of breakpoints.
Since computing this function requires solving at most as many linear programs as there are edge states, this result implies that the overall computation is efficient (i.e., polynomial-time in the input size) whenever the number of edge states is polynomial. 
However, we proceed to show that this is not the case in general.
In fact, even in the single-commodity setting ($|\setCommodities| = 1$) with affine (rather than piecewise affine) cost functions, the number of edge states can be exponential in the input size.
The result follows by adapting a family of instances given by~\citet{GriesbachHoeferKlimm+2022}.

\begin{restatable}{lemma}{LemExponentialBreakpoints}
\label{lem:exponential-breakpoints}
    For every~$k\in \mathbb{N}$, there exists a network with~$\mathcal O(k)$ nodes and~$\mathcal O(k)$ edges, such that the function~$\lambda \mapsto x(\lambda)$ mapping~$\lambda > 0$ to the Wardrop equilibrium in this network has at least~$2^{\Omega(k)}$ edge states.  
\end{restatable}
\newcommand{\proofOfLemExponentialBreakpoints}{
\begin{proof}
    \citet{GriesbachHoeferKlimm+2022} consider a signaling problem in a network with affine cost functions of the form~$c_e^\theta(x)=a_ex+b_e^\theta$, where the offset depends on a state of nature~$\theta\in\Theta$.
    They study how the Wardrop equilibrium behaves as a function of the belief~$\mu$---a probability distribution over the set of states.
    In the case of two different states, the cost function can be rewritten as~$c_e^\lambda(x)=a_ex+b_e^1+\lambda(b_e^2-b_e^1)$, where now~$\lambda\in[0,1]$ indicates the belief (probability) that the second state is realized.
    They construct a family of networks---based on the family of nested Braess networks introduced by \citet{KlimmWarode2021}---in which the equilibrium support changes~$2^{\Omega(k)}$ times as~$\lambda$ varies.
    In their construction, all but one edge have state-independent costs.
    The only edge with state-dependent cost is edge~$e^*=(s,t)$ with~$c_{e^*}^\lambda(x)=\lambda$; note that while this is not strictly increasing in $x$, \citet{GriesbachHoeferKlimm+2022,KlimmWarode2021} argue that the same argument holds when applying small enough positive slopes.
    This suffices to induce~$2^{\Omega(k)}$ distinct equilibrium supports.
\end{proof}
}
\proofOfLemExponentialBreakpoints

\subsection{Binary search for minimal price implementing a feasible budget}
\label{subsec:budget-binary}

As shown in \Cref{lem:exponential-breakpoints}, even for separable and affine travel time functions, and a single externality class, the flow as a function of the price~$\lambda$ may have an exponential number of different regimes. Thus, it may not be efficient to simply trace the evolution of the equilibrium~$x(\lambda)$ to find the minimal price~$\lambda \geq 0$ such that~$G(x(\lambda)) \leq B$ where~$B$ is a feasible budget.
Still, we proceed to show that the minimum price~$\lambda$ to obey a given feasible budget~$B$ can be computed in polynomial time.

\begin{theorem}\label{thm:binary-search}
  Assume that~$|J|=1$,~$\tau$ is separable, piecewise affine, continuous, and strictly increasing, the externality function is separable and constant, and~$B$ is feasible.
  If the input is rational, the minimal price~$\lambda$ with~$G(x(\lambda))\leq B$ can be found in polynomial time.
\end{theorem}
\begin{proof}
  The travel times are piecewise affine and separable, so there are breakpoints~$0 = \BP_{e, 0} < \BP_{e, 1} < \dots < \BP_{e, \numbp} < \BP_{e, \numbp + 1} = \infty$, such that the cost of edge~$e \in E$ is
  \begin{equation*}
  c^{\lambda}_{e} (x_e) = a_{e, \bp} x_e + b_{e, \bp} + g_e \lambda
  \qquad \text{for all } x_e \in [\BP_{e, \bp}, \BP_{e, \bp+1}) \text{ and all } \bp \in \zeroto{\numbp}.
  \end{equation*}

  By first multiplying the demands $d_i$, the budgets $B_j$ and the (finite) breakpoints $\sigma_{e,k}$ by their common denominator, and by dividing the values $a_{e,k}$ by the same value, we may assume without loss of generality that the parameters $d_i$,  $\sigma_{e,k}$ are integer for $i\in I$, $e\in E$ and $k\in\zeroto{K}$; the minimal price implementing the scaled $B$ in the scaled problem coincides with the minimal price implementing the original $B$ in the original problem.
  Similarly, by scaling the cost coefficients $a_{e,k}$, $b_{e,k}$ and $g_e$ and the budget~$B$ by their common denominator, we may also assume without loss that~$a_{e,k}, b_{e,k}, g_e \in \mathbb{Z}$ for all~$e \in E$ and~$k \in \zeroto{K}$; again, this does not change the implementing prices.
  
  Let~$I_0 \subseteq I$ be those commodities for which all paths in~$\mathcal{P}_i$ have the same externality, i.e.,~$\sum_{e \in p} g_e = \sum_{e \in q} g_e$ for all~$p,q \in \mathcal{P}_i$.
  If~$I_0 = I$, all flows have the same externality no matter the choice of~$\lambda$.
  Since~$B$ is feasible, the equilibrium for price~$\lambda = 0$ satisfies~$G(x(0)) \leq B$ and there is nothing left to show. Thus, for the following, we may assume that~$I \setminus I_0 \neq \emptyset$.

  By \Cref{cor:implementability-single-resources}, any feasible budget is implementable, and the corresponding price is in the interval~$[0,\hat{\lambda}+1]$ with the polynomially computable
  $\hat{\lambda} \coloneqq \sum_{e \in E} \tau_e\Bigl(\sum_{i \in I} d_i\Bigr)$.
  By \Cref{lem:piecewise:fixed:computation}, for a given rational price~$\lambda \geq 0$, a corresponding Wardrop equilibrium~$x(\lambda)$ can be computed exactly in polynomial time.
  As the corresponding potential function is strictly convex in aggregate edge loads, the edge flows~$x_e = \sum_{i \in I} x_{i,e}$ are unique and so are the corresponding total externalities~$G(x) = \sum_{e \in E} g_e x_e$.
  We first compute the Wardrop equilibria~$x(0)$ and~$\smash{x(\hat{\lambda} + 1)}$.
  If~$G(x(0)) \leq B$, we set~$\lambda = 0$, and there is nothing left to show.
  By \Cref{thm:monotonicity}, total externality is nonincreasing, so we can perform a binary search to find the minimum price obeying the externality bound.
  
  To this end, we initially set~$\smash{\lambda_{\text{l}} = 0}$, and~$\lambda_{\text{r}} = \hat{\lambda}+1$.
  In each iteration, we compute~$\lambda_{\text{m}} = \frac{\lambda_{\text{r}} + \lambda_{\text{l}}}{2}$
  and compute~$x(\lambda_{\text{m}})$.
  The flow~$x(\lambda_{\text{m}})$ yields a unique corresponding edge state~$\estate = (S, \aparts)$.
  By \Cref{lem:active-supports-interval}, we can compute in polynomial time values~$\smash{\underline{\lambda} \coloneqq \min \Lambda(\estate)}$ and~$\smash{\overline{\lambda} \coloneqq \min\{\lambda_{\text{r}}, \sup \Lambda(\estate)\}}$ such that there is an affine function~$x(\lambda)$ for~$\smash{\lambda \in [\underline{\lambda},\overline{\lambda}]}$ such that~$x(\lambda)$ is a Wardrop equilibrium.
  In particular, the total externality is also unique and affine in the interval~$\smash{[\underline{\lambda},\overline{\lambda}]}$.
  If~$\smash{G(x(\underline{\lambda})) > B}$ and~$G\smash{(x(\overline{\lambda})) \leq B}$, there is a unique~$\smash{\lambda \in (\underline{\lambda},\overline{\lambda}]}$ such that~$G(x(\lambda)) = B$ and this value can be computed explicitly and yields the solution to the problem.
  Otherwise, we proceed as follows: If~$\smash{G(x(\underline{\lambda})) \leq B}$, we set~$\smash{\lambda_{\text{r}} = \underline{\lambda}}$ and recurse.
  If, on the other hand,~$\smash{G(x(\overline{\lambda})) > B}$, we set~$\smash{\lambda_{\text{l}} = \overline{\lambda}}$ and recurse.

  In any case, we keep the invariants, that the interval~$(\lambda_{\text{l}}, \lambda_{\text{r}}]$ contains the minimal $\lambda$ with $G(x(\lambda))\leq B$, and that there is an edge state~$\estate$ such that~$\max \Lambda(\estate) - \min \Lambda(\estate) > 0$ and~$\Lambda(\estate) \subseteq [\lambda_{\text{l}}, \lambda_{\text{r}}]$.
  Both invariants are trivially satisfied at the beginning of the binary search; the first holds inductively after each iteration using \Cref{thm:monotonicity}; the second holds as the (finitely many) intervals in the set  $\{ \Lambda(\estate(\lambda)) \mid \lambda\in\R_{\geq0} \}$ have pairwise disjoint interior and cover $\R_{\geq0}$, and every non-terminating iteration moves an endpoint of our search interval to one of these state-interval endpoints.

  We note that in each iteration the value~$\lambda_{\text{r}} - \lambda_{\text{l}}$ at least halves. To bound the running time of the binary search, it remains to obtain a lower bound on~$\lambda_{\text{r}} - \lambda_{\text{l}}$.
  To this end, we note that there is an edge state~$\mathcal{S}$ with~$\max \Lambda(\mathcal{S}) - \min \Lambda(\mathcal{S}) > 0$ and~$\Lambda(\mathcal{S}) \subseteq [\lambda_{\text{l}}, \lambda_{\text{r}}]$ throughout the iterations of the algorithm. Let~$\lambda^- = \min \Lambda(\mathcal{S})$ and~$\lambda^+ = \sup \Lambda(\mathcal{S})$.
  We proceed to bound~$\lambda^+ - \lambda^-$ from below. If $\lambda^+ = \infty$, there is nothing to show, so we assume that $\lambda^+ = \max \Lambda(\estate)$ is attained. 

  To do so, note that both~$\lambda^+$ and~$\lambda^-$ are optimal solutions (for different objective functions) of the linear program~\eqref{eq:active-support:lp}.
  In particular, these values are attained for basic feasible solutions of the underlying linear program in standard form.
  Let $\eta$ denote the number of variables of this linear program; we have $\eta \leq 2|I||V| + 2|I||E| + 2|E| + 1$.
  At a basic feasible solution, $\eta$ of the constraints are active, so that each basic feasible solution can be obtained as the unique solution of an affine equation system with~$\eta$ variables and~$\eta$ affine equations.
  The absolute values of the coefficients of these affine equations are contained in the set
  \begin{align*}
    M \coloneqq \{0, 1\} \cup \{|\alpha| \mid \alpha\in\{ a_{e,k}, b_{e,k}, \sigma_{e,k}\}, e \in E, k \in \zeroto{K}\} \cup \{g_e \mid e \in E\} \cup \{ d_i \mid i\in I \}.
  \end{align*}
  By Cramer's rule, we obtain~$\lambda^+ = \det (A^+) / \det (B^+)$ and~$\lambda^- = \det(A^-)/\det(B^-)$ for appropriate~$\eta \times \eta$ matrices~$A^+, A^-, B^+$, and~$B^-$ whose absolute coefficients live in~$M$.
  Let~$\mu = \max M$.
  Then, we obtain
  \begin{align*}
    \lambda^+ - \lambda^- = \frac{\det (A^+)}{\det(B^+)} - \frac{\det(A^-)}{\det(B^-)} \geq \frac{1}{|\det(B^+)\det(B^-)|} \geq \frac{1}{ (\eta \mu)^{2\eta} }, 
  \end{align*}
  where for the last inequality we used Hadamard's inequality together with the fact that the~$\ell^2$-norm of each column of the matrices~$B^+$ and~$B^-$ is bounded by~$\eta \mu$.
  Finally, we derive that the number of iterations of the binary search is bounded by
  \begin{align*}
  \big\lceil \log \bigl( (\hat{\lambda}+1) (\eta \mu)^{2\eta} \bigr) \big\rceil = \big\lceil \log (\hat{\lambda}+1) + 2\eta \bigl(\log \eta + \log \mu \bigr) \big\rceil.    
  \end{align*}
  Since~$\eta$, $\log\mu$, and $\log(\hat{\lambda}+1)$ are polynomially bounded in the input, so is the number of iterations.
\end{proof}

\subsection{Application to tradable credit schemes}

The binary search used in \Cref{thm:binary-search} can be applied to the field of tradable credit schemes introduced in~\citet{Yang2011}.
In our notation, a \emph{tradable credit scheme} is a pair~$(B, g)$ where $B\in\R_{\geq0}$ is the number of issued credits, and $g_e\in\R_{\geq0}$ is the credit charge for traversing an edge~$e\in E$.
A tradable credit scheme $(B, g)$ is \emph{feasible} if there is a feasible flow $x$ with $G(x)\coloneqq \sum_{e\in E} g_e\cdot x_e \leq B$, or equivalently if $B$ is a feasible budget with respect to{} the externality function~$g$.
Furthermore, a flow--price pair $(x, \lambda)$ with~$x$ feasible and $\lambda \geq 0$ is a \emph{market equilibrium} with respect to a tradable credit scheme~$(B, g)$ if $x\in \WE(\lambda)$, $G(x) \leq B$, and~$\lambda \cdot (G(x) - B) = 0$.
In this case, $\lambda$ is called a \emph{market equilibrium price}.

As discussed by \citet{Yang2011}, the market equilibrium price may not be unique.
In the following, we characterize the set of market equilibrium prices for a given tradable credit scheme~$(B, g)$ and travel time functions $\tau$ with a convex potential function as a closed interval.

\begin{lemma}\label{lem:characterization-market-equilibrium-price}
    Let $(B, g)$ be a feasible tradable credit scheme, and let~$\tau$ be travel time functions with a convex potential function.
    Then, a price $\lambda\in\R_{\geq0}$ is a market equilibrium price if and only if $\lambda\in[\underline{\lambda}, \overline{\lambda}]$ with
\begin{align*}
    \underline{\lambda} &\coloneqq \min\{\lambda\in\R_{\geq0} \mid \exists x\in \WE(\lambda): G(x) \leq B\}, \text{ and }\\
    \overline{\lambda} &\coloneqq \sup(\{\lambda\in\R_{\geq0} \mid \exists x \in \WE(\lambda): G(x) \geq B\} \cup \{ 0 \} ).
\end{align*}
\end{lemma}
\begin{proof}
    We first prove the following property:
    \begin{claim}
        For closed $S\subseteq\R_{\geq0}$, the set $\mathcal{L}(S)\coloneqq \{ \lambda\in\R_{\geq0} \mid \exists x\in\WE(\lambda): G(x) \in S \}$ is closed.
    \end{claim}
    \begin{subproof}
        Let $(\lambda_i)_{i\in\N}$ be a sequence in $\mathcal{L}(S)$ converging to some $\lambda^*\in\R_{\geq0}$.
        By definition, there exists a sequence $(x_i)_{i\in\N}$ with $x_i\in\WE(\lambda_i)$ and $G(x_i)\in S$ for all $i$.
        Since the set of feasible flows is compact, there exists a subsequence $(x_{i_j})_{j\in\N}$ that converges to some feasible flow $x^*$.
        By \Cref{prop:wardrop-iff-minimizes-potential}, every $x_i$ minimizes the convex potential function $\Phi^{\lambda_i}$ over $\feasFlows$.
        As the map $(\lambda, x)\mapsto \Phi^\lambda(x)$ is continuous,~$x^*$ minimizes $\Phi^{\lambda^*}$ over~$\feasFlows$, thus $x^*\in\WE(\lambda^*)$.
        Finally, since $G$ is continuous and $S$ is closed, we have $G(x^*) = \lim_{j\to\infty} G(x_{i_j}) \in S$, and thus $\lambda^*\in \mathcal{L}(S)$.
    \end{subproof}
    
    Note that $\underline{\lambda} = \min\mathcal{L}([0, B])$ is well-defined, since $B$ is feasible (and by \Cref{thm:implementability-for-potential-travel-time} implementable), and thus $\mathcal{L}([0, B])$ is nonempty, bounded from below by $0$ and closed.
      \par
      Now, assume that $(\lambda, x)$ is a market equilibrium.
      Clearly, $\underline{\lambda} \leq \lambda$ since $G(x) \leq B$ and $x\in \WE(\lambda)$.
      If we assume $\lambda > \overline{\lambda}$, then by definition of $\overline{\lambda}$, we have $G(x) < B$.
      By the market equilibrium condition, this implies $\lambda = 0$, a contradiction to $\lambda > \overline{\lambda}\geq0$.
      Thus, $\lambda\in[\underline{\lambda}, \overline{\lambda}]$.
      
      For the other direction, assume that $\lambda \in [\underline{\lambda}, \overline{\lambda}]$.
      If $\lambda = 0$, choose $x\in\WE(0)$ with $G(x)\leq B$, which exists because $\underline\lambda=0$; then $(x,0)$ is a market equilibrium.
      Thus, we may assume $\lambda > 0$.
      
      If $\lambda = \underline{\lambda} > 0$, then there exists a flow $y\in\WE(\lambda)$ with $G(y) \leq B$ (by definition of $\underline{\lambda}$).
      Furthermore, for a sequence $(\lambda_i)_i$ converging to $\lambda$ from below, we know $\lambda_i\in\mathcal{L}([B, \infty))$ by the monotonicity (\Cref{thm:monotonicity}), and thus $\lambda\in\mathcal{L}([B, \infty))$.
      This means, there exists a flow $z\in\WE(\lambda)$ with $G(z) \geq B$.
      Since $\WE(\lambda)$ is convex{} (\Cref{prop:wardrop-iff-minimizes-potential}), and $G$ is linear, there is $x\in\WE(\lambda)$ with $G(x) = B$.
      Thus, $(\lambda,x)$ is a market equilibrium.

      Similarly, if $\lambda = \overline{\lambda} > 0$, we know that $\overline{\lambda}$ is finite, and hence, $\overline{\lambda} = \max\mathcal{L}([B, \infty))$.
      A sequence $(\lambda_i)_i$ converging to $\lambda$ from above fulfills $\lambda_i\in\mathcal{L}([0, B])$ by monotonicity, and thus $\lambda\in\mathcal{L}([0, B])$.
      Again, by convexity, there exists a flow $x\in\WE(\lambda)$ with $G(x) = B$ such that $(\lambda, x)$ is a market equilibrium.

      The remaining case is $\underline{\lambda} < \lambda < \overline{\lambda}$.
      Let $x$ be any flow in $\WE(\lambda)$.
      By definition of $\underline{\lambda}$, there exists a flow $y\in\WE(\underline{\lambda})$ with $G(y) \leq B$.
      By \Cref{thm:monotonicity}, $G(x)\leq G(y)\leq B$.
      Further, since $\lambda < \overline{\lambda}$, there exists some $\lambda'\in (\lambda, \overline{\lambda}]$ and a flow $\overline{x}\in\WE(\lambda')$ with $G(\overline{x}) \geq B$.
      Applying \Cref{thm:monotonicity} now to $\lambda$ and $\lambda'$, it follows that~$G(x) \geq G(\overline{x})\geq B$.
      Summing up, we have $G(x) = B$, and thus $(\lambda, x)$ is a market equilibrium.
\end{proof}

  Restricting ourselves again to separable, piecewise affine travel time functions, we can apply the theory developed in this section to compute the whole set of market equilibrium prices in polynomial time.

\begin{theorem}
  \label{cor:tradable-credit-scheme-market-price}
    Given a feasible tradable credit scheme~$(B, g)$ and separable, piecewise affine, continuous, and strictly increasing travel time functions~$\tau$, a market equilibrium $(x, \lambda)$ can be computed in polynomial time.
    The market equilibrium prices form a closed interval, also computable in polynomial time.
\end{theorem}
\begin{proof}
    For the special case of piecewise affine travel time functions, the total externality $G(x)$ is unique among all Wardrop equilibria $x\in\WE(\lambda)$ for a fixed price $\lambda$.
    The lower bound $\underline{\lambda}$ of the interval defined in \Cref{lem:characterization-market-equilibrium-price} can be obtained in polynomial time by applying \Cref{thm:binary-search}.
    A corresponding market equilibrium flow $\underline{x}$ can be obtained from \Cref{lem:piecewise:fixed:computation}.

    If $\smash{G(\underline{x})< B}$, then both $\smash{\underline{\lambda}}$ and the upper bound~$\smash{\overline{\lambda}}$ of the desired interval are zero; thus, no further work is required.
    Otherwise, we can compute the minimum feasible budget $B_{\min}$ in polynomial time (by computing an equilibrium for any price higher than the one given in \Cref{cor:implementability-single-resources}).
    Then, the upper bound $\smash{\overline{\lambda}}$ is infinite if and only if $B = B_{\min}$.
    Again, if $\smash{\overline{\lambda}}$ is infinite, we are done.
    Otherwise, we can adapt the binary search illustrated in the proof of \Cref{thm:binary-search} to find~$\smash{\overline{\lambda}}$ as the maximum price $\lambda$ with $G(x(\lambda)) = B$ on the interval $\smash{[\underline{\lambda}, \hat{\lambda}+1]}$.
    At equality, we continue towards larger prices; the same denominator bound and rational reconstruction guarantee exact polynomial-time termination.
\end{proof}

\section{Computational study}\label{sec:comp-study}

\newcommand{\figureNonConvexBraess}{%
\begin{figure}[ht]
      \centering
      \begin{tikzpicture}
            \tikzset{
            solidnode/.style={
                circle, fill, inner sep=0ex, outer sep=0ex, minimum size=1ex
            }}
    
            \node[solidnode, label=left:{$s$}] (s) at (0,0) {};
            \node[solidnode, label=above:{$u$}] (v) at (3,1) {};
            \node[solidnode, label=below:{$v$}] (w) at (3,-1) {};
            \node[solidnode, label=below:{$w$}] (t) at (6,0) {};
            \node[solidnode, label=right:{$t$}] (z) at (9,0) {};
    
            \draw[-{stealth}]
                (s) edge node[midway, sloped, above] {\footnotesize$1 + \frac{x_e^4}{4}+\lambda\cdot 0$} (v)
                (s) edge node[midway, sloped, below] {\footnotesize$17 + \lambda \cdot 0$}  (w)
                (v) edge node[midway, right] {\footnotesize$0 + \lambda \cdot 1$} (w)
                (v) edge node[midway, sloped, above] {\footnotesize$17 + \lambda \cdot 0$}  (t)
                (w) edge node[midway, sloped, below] {\footnotesize$1 + \frac{x_e^4}{4}+\lambda\cdot 0$}  (t)
                (t) edge[bend left] node[midway, above] {\footnotesize$16 + \lambda\cdot 0$} (z)
                (t) edge[bend right] node[midway, below] {\footnotesize$1 + \lambda \cdot 1$} (z);
    
            \draw[decorate, decoration={snake, amplitude=0.5mm, segment length=4mm}, -{stealth}]
                (-2.5,0) -- (-.75, 0) node[midway, above] {\footnotesize$d=4$};
      \end{tikzpicture}
      \vspace{1em}
      \begingroup
      \renewcommand{\plotUnits}[1]{}
      \pgfplotsset{y filter/.expression={y/1000}}
      \plotInstanceTikz{non-convex-braess-carbon-pricing}{}
      \endgroup
      \vspace{-1em}
      \caption{Carbon pricing in a Braess-inspired network (top) with edges labelled as~$\tau_e(x_e) + \lambda\cdot g_e(x_e)$. The resulting emission (bottom left) is neither convex nor concave, and the travel time (bottom right) is non-monotone.
      The dashed line indicates the minimum feasible emissions.}
      \label{fig:braess-comp-study}
    \end{figure}
}

\newcommand{\carbonPricingNetworkRow}[2]{%
  \nextgroupplot[ylabel={\begin{tabular}{c}Total Emission\\\tiny (in km)\end{tabular}}, extra description/.code={%
    \node[rotate=90, font=\linespread{1}\fontsize{9}{10.35}\bfseries\selectfont, align=center,
      text width=0.11\textheight]
      at ([xshift=-2.4cm]axis description cs:0,0.5) {#2};
  }]
    \addplot[-] table[x=price, y=total_consumption, col sep=comma] {data/#1-carbon-pricing.csv};
    \carbonPricingLowerBound{#1-carbon-pricing}{1}
  \nextgroupplot[ylabel={\begin{tabular}{c}Total Travel\\Time \tiny (in min)\end{tabular}}]
    \addplot[-] table[x=price, y=total_travel_time, col sep=comma] {data/#1-carbon-pricing.csv};
}

\newcommand{\figureCarbonPricingNetworks}{%
  \begin{figure}[p]
    \centering
    \begin{tikzpicture}
      \begin{groupplot}[
        group style={group size=2 by 6, horizontal sep=2cm, vertical sep=1cm,
          xlabels at=edge bottom, xticklabels at=all},
        scale only axis,
        width=0.35\textwidth,
        height=0.11\textheight,
        grid=major,
        tick label style={font=\scriptsize},
        label style={font=\footnotesize},
        xlabel={\begin{tabular}{c}Price\\[-0.5ex]\tiny (in min/km)\end{tabular}},
        max space between ticks=35pt,
      ]
        \carbonPricingNetworkRow{anaheim}{Anaheim}
        \carbonPricingNetworkRow{berlin-friedrichshain}{Berlin\\Friedrichshain}
        \carbonPricingNetworkRow{berlin-tiergarten}{Berlin\\Tiergarten}
        \carbonPricingNetworkRow{berlin-prenzlauerberg-center}{Berlin\\\mbox{Prenzlauer Berg}}
        \carbonPricingNetworkRow{winnipeg}{Winnipeg}
        \carbonPricingNetworkRow{sioux-falls}{Sioux Falls}
      \end{groupplot}
    \end{tikzpicture}
    \caption{Carbon pricing in six more real-world networks: total emissions (left) and total travel time (right).
    Dashed lines indicate the minimum feasible emissions.}
    \label{fig:carbon-pricing-networks}
    \label{fig:anaheim-carbon-pricing}
    \label{fig:berlin-friedrichshain-carbon-pricing}
    \label{fig:berlin-tiergarten-carbon-pricing}
    \label{fig:berlin-prenzlauerberg-center-carbon-pricing}
    \label{fig:winnipeg-carbon-pricing}
    \label{fig:sioux-falls-carbon-pricing}
  \end{figure}
}

\newlength{\plotheight}
\newlength{\plotwidth}
\setlength{\plotheight}{0.25\textwidth}
\setlength{\plotwidth}{0.4\textwidth}

In this computational study, we investigate the effects of externality pricing in real-world networks through two scenarios:
carbon pricing and cordon pricing.{}
We examine how prices affect total emissions and travel time and, for cordon pricing, city-center emissions and total flow permitted to the city center.

    As discussed in \Cref{sec:intro}, for \emph{carbon pricing}, each agent pays a price for each unit of emission they produce.
    In this study, we are given a constant CO\sub2 emission factor~$g_e\in\R_{\geq0}$ for each edge~$e$.
    Similarly, for \emph{cordon pricing}, we are given a subset of edges~$C\subseteq E$ representing the city center for which a fee must be paid to be permitted.
    This is modelled with path-based externality values by setting~$g_{i,p} \coloneqq 1$, if~$p$ contains an edge in $C$, and~$g_{i,p}\coloneqq 0$, otherwise.

Finally, we instantiate the travel time functions~$\tau_e$ using the BPR (Bureau of Public Roads) function
\begin{equation}\label{eq:bpr}
  \tau_e(x_e) \coloneqq \alpha_e\cdot (1+\beta_e\cdot (\nicefrac{x_e}{\gamma_e})^4)
\end{equation}
where~$\alpha_e\geq 0$ is the free-flow travel time,~$\gamma_e>0$ is the capacity and~$\beta_e$ is a congestion factor.
    This means we can minimize the BMW potential to find an equilibrium with respect to the parametrized cost function~$c^\lambda$.
We discretize the range of prices~$\lambda$ and determine an approximate minimum of the BMW potential using a simple Frank--Wolfe type algorithm as first proposed in \citet{Dafermos69}.
For each plot, we sample~$101$ prices~$\lambda$, terminating Frank--Wolfe when the relative gap falls below~$0.0001\%$.

We implemented the Frank--Wolfe algorithm in the Rust programming language.
We use an adapted A* search algorithm, compatible with the path-based cordon-pricing formulation, to quickly solve the shortest path problems arising when solving the linearized subproblems during the Frank--Wolfe iteration.
The implementation is available at \citep{Markl2026}.

  \subsection{Data}
The networks we use for this study are included in the dataset in \citep{transportation-networks}.
  Note that we use the length of an edge as a proxy for the carbon emission values.
  While the elevation and the type of roads are relevant factors for the real-world emissions, they are not considered here.
  Similarly, to compare prices with travel times, a value of time factor needs to be considered.
  As we do not focus on the correct choice of this parameter, we simply consider prices to be measured in time units.
\begin{table}[tbph]
  \caption{List of considered real-world networks. Here,~$T_{\text{carbon}}$ is the total time in seconds needed for computing the equilibria for all 101 samples of~$\lambda$ in the carbon pricing scenario.}\label{table:networks}
  \centering
  \rowcolors{2}{black!5}{white}
  \footnotesize
  \begin{tabular}{l  r  r  r  r  r r r}
    \toprule
    Network Name &~$|V|$ &~$|E|$ &~$|I|$ & \shortstack[r]{Total \\ Demand} & \shortstack[r]{Time Unit \\ (in min)} & \shortstack[r]{Distance Unit\\ (in km)} & \shortstack[r]{ $T_{\text{carbon}}$ \\ (in s) } \\
    \midrule
    Anaheim &  416 & 914 & 1{,}406 & 104{,}694 & 1\hphantom{.00000} &  0.0003048  & 15.7 \\ 
    Berlin Center & 12{,}981 & \quad 28{,}376 & \quad 49{,}688 & 168{,}222 & 0.01667 & 0.001\hphantom{0000} & 853.4 \\
    Berlin Friedrichshain & 224 & 523 & 506 & 11{,}205 & 0.01667 & 0.001\hphantom{0000} & 7.5 \\
    Berlin Prenzlauer Berg & 352 & 749 & 1{,}406 & 16{,}660&  0.01667 & 0.001\hphantom{0000} & 79.0 \\
    Berlin Tiergarten & 359 & 766 & 644 & 10{,}755 & 0.01667 & 0.001\hphantom{0000} & 14.6 \\
    Sioux Falls & 24 & 76 & 528 & 360{,}600 & 0.6\hphantom{0000} & 1.60934\hphantom{00} & 134.2 \\
    Winnipeg & 1{,}040 & 2{,}836 & 4{,}344 & 64{,}775 & 1\hphantom{.00000} & 1\hphantom{.0000000} & 1175.8 \\
    \bottomrule
  \end{tabular}
\end{table}

  \Cref{table:networks} lists the networks we considered in this computational study together with some characteristic parameters.
  Not all datasets from \citep{transportation-networks} come with units for distance and time measurements.
  Thus, the table also contains the conversion factors we used for this study.
  For cordon pricing, we augmented the Berlin Center network with information from OpenStreetMap~\citep{OpenStreetMap} by determining which nodes lie inside the Ringbahn, a circular rail route around Berlin's inner city area.
  Note that a segment of the A100 highway runs parallel to the Ringbahn, and in some cases switches between the inner and outer side of the Ringbahn.
  To obtain a realistic scenario, we manually placed the A100 edges outside the cordon.
  \Cref{fig:berlin} depicts this network and its inner city area.  

\newcommand{\figureBerlinCordon}{
\begin{figure}[htbp]
  \centering
    \includegraphics[width=11cm]{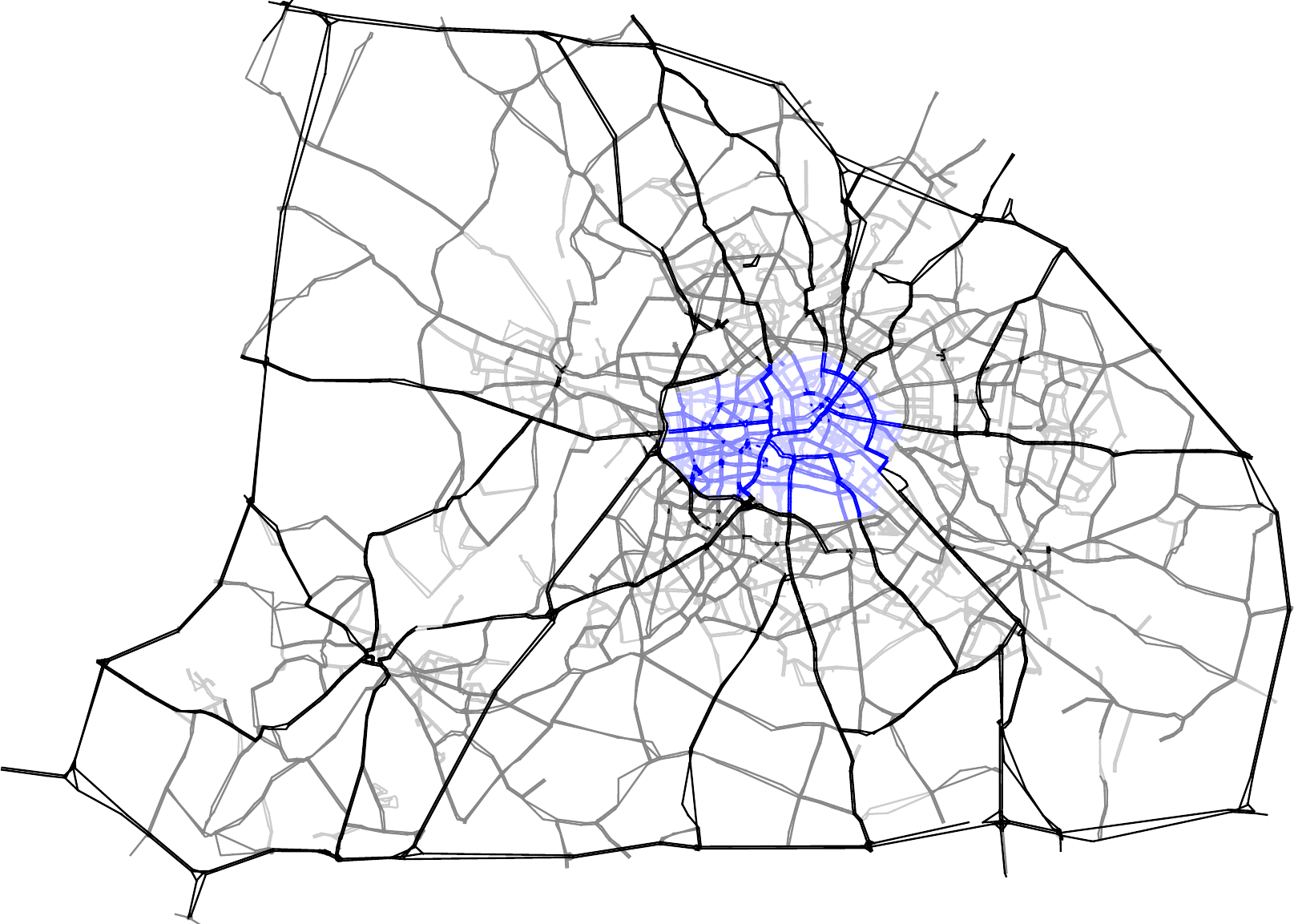}
  \caption{The Berlin Center network. The inner city edges are highlighted in blue, the others in black. The intensity of the color corresponds to the capacity of the edge.}
  \label{fig:berlin}
\end{figure}
}
\figureBerlinCordon

\newcommand{\CompStudyConversionFactors}{
    Not all datasets from \citep{transportation-networks} come with units for distance and time measurements.
    We collect the conversion factors we used for this study in \Cref{table:conversion-factors}.
    
    \begin{table}[htbp]
      \caption{The conversion factors for the units used in the considered real-world networks from \citet{transportation-networks}.}\label{table:conversion-factors}
      \centering
      \tiny
      \rowcolors{2}{black!5}{white}
      \begin{tabular}{l  r  r}
        \toprule
        Network Name & Time Unit (min) & \quad Distance Unit (km) \\
        \midrule
        Anaheim & 1\hphantom{.00000} &  0.0003048  \\ 
        Berlin Center & 0.01667 & 0.001\hphantom{0000} \\
        Berlin Friedrichshain & 0.01667 & 0.001\hphantom{0000} \\
        Berlin Prenzlauer Berg&  0.01667 & 0.001\hphantom{0000}  \\
        Berlin Tiergarten & 0.01667 & 0.001\hphantom{0000} \\
        Sioux Falls & 0.6\hphantom{0000} & 1.60934\hphantom{00} \\
        Winnipeg & 1\hphantom{.00000} & 1\hphantom{.0000000} \\
        \bottomrule
      \end{tabular}
    \end{table}
}

  \subsection{Carbon pricing}

  We first examine the results for the application of carbon pricing.
{}

For the Berlin Center network, \Cref{fig:berlin-center-carbon-pricing} illustrates the effect of the carbon price on the total carbon emission, i.e.,~$G(x(\lambda))=\sum_{e\in E} g_e\cdot x(\lambda)_e$, and the total travel time, i.e.,~$\sum_{(i,p)\in \mathcal P}\tau_{i,p}(x(\lambda))\cdot x(\lambda)_{i,p}$, where~$x(\lambda)$ denotes the Wardrop equilibrium computed for price~$\lambda$. 
The total emission is measured in equivalent kilometers driven, and the price is measured in time units representing the inverse value-of-time factor.
The figure demonstrates that the total emission is monotonically decreasing, as guaranteed by \Cref{thm:monotonicity}, and exhibits a convex shape, with the most significant reduction in total emission occurring at lower price values.
Conversely, the total travel time displays a seemingly concave curve.
  Note that for this network, the minimal feasible budget is~$6.15\cdot 10^5$ km, which is almost attained within the depicted price interval.
As can be seen, the carbon price can reduce total emissions by up to~$4.2$\%.

\plotInstance{berlin-center-carbon-pricing}{Carbon pricing in the Berlin Center network.
The dashed line indicates the minimum feasible emissions.}
\figureCarbonPricingNetworks

    The results of the carbon-pricing study for the other real-world networks are depicted in \Cref{fig:carbon-pricing-networks}.
    They suggest that in real-world networks (a) the total carbon emission curve behaves in a convex shape, where the total emission decreases faster at low prices, and (b) the total travel time usually exhibits a concave, monotonically increasing curve.

    An interesting exception to this observation is the Anaheim network.
    In this network, the total travel time decreases at low prices and remains below the zero-price travel time until the price reaches 3.8 min/km.

    \begin{example}
        Note that for arbitrary networks, the total emission is not always convex, nor is the total travel time always monotonically increasing.
        Consider, for example, the adaptation of Braess' network in \Cref{fig:braess-comp-study}.
        For this network, low prices ensure that the edge $(u,v)$, which traditionally induces the Braess paradox, lies on a shortest path.
        Increasing the price on this edge reduces flow on it, thereby decreasing total travel time.
        For prices higher than $15$, the two parallel edges from $w$ to $t$ then lead to an increase in total travel time, since the upper, slower edge with $g_e=0$ is more attractive than the faster edge with $g_e = 1$.
      \figureNonConvexBraess
    \end{example}

  \subsection{Cordon pricing}

We analyze the cordon pricing strategy as implemented in London, where every driver has to pay a fee to be permitted to drive in the city center, and the fee is independent of the number of times the driver (re)enters the city center.
If we denote the set of edges lying in the city center by $C$, this scenario is modelled using a path-based externality $g_{i,p}$ with $g_{i,p} = 1$, if $p$ uses an edge $e$ in $C$, and $g_{i,p} =0$, otherwise.
In this scenario, \Cref{thm:monotonicity} ensures that the total flow using at least one edge in the city center, i.e., $\sum_{i,p} g_{i,p} \cdot x_{i,p}$, decreases monotonically with the price~$\lambda$ charged for the city center permit.

\Cref{fig:berlin-center-permit-based-cordon} presents the effects of cordon pricing on the Berlin Center network.
A primary goal of cordon pricing in practice is to reduce emissions within the city center.
For the Berlin Center network, the number of users driving in the city center can be reduced by up to 2.4\%; furthermore, emissions in the city center can be reduced by up to 1.4\%.
It is, however, noticeable that the minimum feasible emission in the city center is much lower (even with fixed demands) than what is achievable through means of cordon pricing.
This implies that in this network, a big share of carbon emission in the city center stems from users whose origin or destination lies in the city center.
Finally, note that as the price increases, the total emission in the entire network increases slightly, which is a consequence of drivers taking longer, less efficient routes to avoid the fee.
These longer routes also lead to an increase in total travel time of up to 0.2\%.

\newcommand*{\plotCordonInstance}[2]{
  \begin{figure}[ht]
    \centering
    {\footnotesize
    \begin{tikzpicture}
      \begin{groupplot}[
        group style={group size=2 by 2, horizontal sep=3cm, vertical sep=1cm},
        height=\plotheight, width=\plotwidth,
        grid=major,
        y tick label style={
            /pgf/number format/.cd,
            fixed,
            fixed zerofill,
            precision=3,
            /tikz/.cd
          }]
        \nextgroupplot[
          ylabel={\begin{tabular}{c} Total Emission\\ \tiny (in km)\end{tabular}},
        ]
            \addplot[-] table [x=price, y=total_consumption, col sep=comma] {data/#1.csv};
            \pgfplotstableread[col sep=comma]{data/berlin-center-min-emission.csv}\minimumEmissionTable
            \pgfplotstablegetelem{0}{total_consumption}\of\minimumEmissionTable
            \edef\minimumEmission{\pgfplotsretval}%
            \addplot[black, dashed, no marks, forget plot]
              table[x=price, y expr=\minimumEmission, col sep=comma] {data/#1.csv};

        \nextgroupplot[
          ylabel={\begin{tabular}{c}Total Travel \\ Time \tiny (in min)\end{tabular}},
          y tick label style={
            /pgf/number format/.cd,
            fixed,
            fixed zerofill,
            precision=3,
            /tikz/.cd
          },
        ]

          \addplot[-] table [x=price, y=total_travel_time, col sep=comma] {data/#1.csv};
        
        \nextgroupplot[
          xlabel={\begin{tabular}{c}Price\\[-0.5ex] \tiny (in min/permit)\end{tabular}}, ylabel={\begin{tabular}{c}Total Emission\\ in City \tiny (in km)\end{tabular}},
        ]
          \addplot[-] table [x=price, y=total_consumption_inside, col sep=comma] {data/#1.csv};
          \pgfplotstableread[col sep=comma]{data/berlin-center-min-emission-inside.csv}\minimumEmissionInsideTable
          \pgfplotstablegetelem{0}{total_consumption_inside}\of\minimumEmissionInsideTable
          \edef\minimumEmissionInside{\pgfplotsretval}%
          \addplot[black, dashed, no marks, forget plot]
            table[x=price, y expr=\minimumEmissionInside, col sep=comma] {data/#1.csv};

        \nextgroupplot[xlabel={\begin{tabular}{c}Price\\[-0.5ex]\tiny (in min/permit)\end{tabular}}, ylabel={\begin{tabular}{c}Flow Permitted to \\ City \tiny (in \#vehicles)\end{tabular}}]
            \addplot[-] table [x=price, y=total_permit_flow, col sep=comma] {data/#1.csv};
            \pgfplotstableread[col sep=comma]{data/berlin-center-min-permits.csv}\minimumPermitTable
            \pgfplotstablegetelem{0}{total_permit_flow}\of\minimumPermitTable
            \edef\minimumPermitFlow{\pgfplotsretval}%
            \addplot[black, dashed, no marks, forget plot]
              table[x=price, y expr=\minimumPermitFlow, col sep=comma] {data/#1.csv};

      \end{groupplot}
    \end{tikzpicture}}
    \vspace{-1em}
    \caption{#2}\label{fig:#1}
  \end{figure}
}
\plotCordonInstance{berlin-center-permit-based-cordon}{Cordon pricing in the Berlin Center network.
Dashed lines mark the minima over all feasible flows.}

\newcommand*{\plotInstanceBadLabel}[3][]{
  \begin{figure}[#1]
    \begin{tikzpicture}
      \begin{axis}[width=0.3\textwidth, xlabel=Price (in min/100km), ylabel=Total Emission (in km), x label style={at={(axis description cs:0.5,-0.15)},anchor=north}]
        \addplot[-] table [x expr=\thisrow{price}*100, y=total_consumption, col sep=comma] {data/#2.csv};
      \end{axis}
    \end{tikzpicture}\hspace{0.1\textwidth}
    \begin{tikzpicture}
      \begin{axis}[width=0.3\textwidth, xlabel=Price (in min/100km), ylabel=Total Travel Time (in min), x label style={at={(axis description cs:0.5,-0.15)},anchor=north}]
        \addplot[-] table [x expr=\thisrow{price}*100, y=total_travel_time, col sep=comma] {data/#2.csv};
      \end{axis}
    \end{tikzpicture}
    \caption{#3}\label{fig:carbon-pricing:#2}
  \end{figure}
}

\section{Conclusion}\label{sec:conclusion}

We studied how a common price per unit of an externality can shape Wardrop
equilibria in traffic networks.
Our theoretical results identify which flows and externality budgets can be implemented, and distinguish between the existence of an equilibrium that respects a budget and the stronger guarantee that every equilibrium respects it.
The results provide both structural guarantees and computational
tools for choosing externality prices.

Our computational study illustrates how these results translate into network
effects.
Across most of the real-world instances, a substantial share of the
attainable emission reduction occurs already at comparatively low prices,
whereas further price increases yield progressively smaller reductions.
At the same time, total travel time generally increases more gradually.
This pattern suggests that moderate prices may capture a large part of the available emission benefit while avoiding the largest travel-time costs.
The Anaheim instance also shows that this trade-off need not be monotone: a low price can initially reduce total travel time by redirecting traffic away from inefficient routes.
The cordon-pricing experiment highlights a different limitation of the instrument.
Although the charge reduces traffic and emissions inside the cordon, some traffic is diverted to longer routes, slightly increasing travel time and emissions in the network as a whole.
Moreover, a cordon charge cannot eliminate trips whose origin or destination lies inside the charged area.
These observations emphasize that the location and scope of a pricing instrument matter as much as its level, and that local emission improvements should be evaluated together with network-wide displacement effects.

The framework can also be applied to games with elastic road demand by adding,
for each commodity, an outside option with a prescribed generalized cost.
This outside option may represent public transport (or forgoing the trip),
so that the fixed population chooses between road paths and the outside option and the
amount of road traffic becomes endogenous.
In this representation, the game retains the form studied in this paper,
and our implementability and pricing results continue to apply.
Elastic-demand traffic equilibria and their multimodal extensions have been studied, for example, by
\citet{FlorianNguyen1974} and \citet{Dafermos1982}.
Characterizing these games, especially when public-transport costs, capacities,
crowding, or emissions depend on participation, and evaluating them empirically are left for future research.

Several further questions remain open.
It would be useful to develop pricing rules that are robust to equilibrium selection and to uncertain travel-time
and emission parameters, and to extend the efficient algorithms beyond separable piecewise-affine travel times.
Dynamic traffic and distributional effects of pricing are natural next steps.
On the empirical side, the computational study uses edge length as a proxy for emissions and holds travel demand fixed.
Calibrating the model with speed-dependent emission functions, observed demand elasticities, mode-choice data, and public-transport capacity would clarify which qualitative patterns persist in policy-relevant settings.

  \clearpage
  \bibliography{main,master-bib}

  \clearpage

\end{document}